\documentclass[lettersize,journal]{IEEEtran}
\usepackage{amsmath,amsfonts}
\usepackage{algorithm}
\usepackage{array}
\usepackage[caption=false,font=normalsize,labelfont=sf,textfont=sf]{subfig}
\usepackage{textcomp}
\usepackage{url}
\usepackage{verbatim}
\usepackage{graphicx}
\usepackage{cite}
\usepackage{multirow}
\usepackage{amsmath, amsfonts, amsthm, amssymb, bbm}
\usepackage[utf8]{inputenc}
\usepackage{bm}
\usepackage[dvipsnames]{xcolor}
\usepackage[colorlinks,citecolor=OliveGreen,urlcolor=blue,linkcolor=blue,bookmarks=true]{hyperref}      
\usepackage[nameinlink,capitalize]{cleveref}

\usepackage{dsfont}

\usepackage{physics}

\newcommand{\bb}{\mathbb}
\newcommand{\mc}{\mathcal}
\newcommand{\row}{\operatorname{row}}

\newcommand{\res}[2]{\left. #1 \right|_{#2}}
\newcommand{\rk}{\operatorname{rank}}
\def\nn{\nonumber}

\usepackage{nicefrac}
\newcommand{\fig}[1]{Fig.~\ref{fig:#1}}
\newcommand{\secref}[1]{Sec.~\ref{sec:#1}}

\def\vcentcolon{\mathrel{\mathop\ordinarycolon}}
\begingroup \catcode`\:=\active
\lowercase{\endgroup
	\let :\vcentcolon
}

\newtheorem{theorem}{Theorem}[section]
\newtheorem*{theorem*}{Theorem}
\newtheorem{definition}[theorem]{Definition}

\newtheorem{proposition}[theorem]{Proposition}
\newtheorem{lemma}[theorem]{Lemma}

\def\iden{\mathds{1}}

\def\iden{\mathds{1}}

\usepackage{quantikz}
\usepackage{algorithm}
\usepackage{algpseudocode}

\begin{document}
	
\title{Power and Limitations of Linear Programming Decoder for Quantum LDPC Codes}

\author{
	\IEEEauthorblockN{Shouzhen Gu\IEEEauthorrefmark{1} and Mehdi Soleimanifar\IEEEauthorrefmark{2}}\\
	\IEEEauthorblockA{\IEEEauthorrefmark{1}Yale Quantum Institute \& Department of Applied Physics, Yale University, New Haven, CT, USA}\\
	\IEEEauthorblockA{\IEEEauthorrefmark{2}California Institute of Technology, Pasadena, CA, USA}
}

\maketitle

\begin{abstract}
    Decoding quantum error-correcting codes is a key challenge in enabling fault-tolerant quantum computation. In the classical setting, linear programming (LP) decoders offer provable performance guarantees and can leverage fast practical optimization algorithms. Although LP decoders have been proposed for quantum codes, their performance and limitations remain relatively underexplored.
    In this work, we uncover a key limitation of LP decoding for quantum low-density parity-check (LDPC) codes: certain constant-weight error patterns lead to ambiguous fractional solutions that cannot be resolved through independent rounding. To address this issue, we incorporate a post-processing technique known as ordered statistics decoding (OSD), which significantly enhances LP decoding performance in practice. Our results show that LP decoding, when augmented with OSD, can outperform belief propagation with the same post-processing for intermediate code sizes of up to hundreds of qubits. These findings suggest that LP-based decoders, equipped with effective post-processing, offer a promising approach for decoding near-term quantum LDPC codes.
\end{abstract}

\begin{IEEEkeywords}
Decoding algorithm, linear programming, maximum likelihood decoding, ordered statistics decoding, quantum low-density parity-check codes.
\end{IEEEkeywords}

\section{Introduction}
Quantum error correction (QEC) is widely believed to be necessary to protect quantum information from the inevitable noise and imperfections that arise in quantum systems~\cite{Shor95, Preskill98, NielsenChuang10}. 
Effective QEC codes should have a large distance and high rate so that they can tolerate a significant number of errors while maintaining low overhead.
Another desirable feature is that the stabilizer measurements used to obtain information about errors have low weight, allowing for efficient and fault-tolerant syndrome extraction.
Quantum low-density parity-check (LDPC) codes~\cite{Breuckmann2019LDPC} are families of codes that satisfy this sparsity property, and they hold the potential for enabling large-scale quantum computation with low overhead~\cite{Gottesman2014, bravyi2024high, xu2024constant, NguyenPattison24overhead, YamasakiKoashi24}.
Recent constructions of quantum LDPC codes have led to improved parameters both in the asymptotic regime~\cite{Tillich2014LDPC, HastingsHaahODonnell2021, PKalmostlinear, Breuckmann2021LDPC, Panteleev2022Goodcode, leverrier2022quantum, Dinur2023goodCode} and for finite code sizes~\cite{bravyi2024high, xu2024constant}.

To make LDPC codes practically useful, they must be paired with efficient \emph{decoding algorithms} capable of determining errors that have affected the quantum state using information from the stabilizer measurements.
Computationally efficient and provably correct decoders have been introduced for certain families of LDPC codes, including quantum expander codes~\cite{leverrier2015quantum, FawziGrospellierLeverrier18} and asymptotically good codes~\cite{Gu2023decoder, Leverrier2024Decoder, Dinur2023goodCode, gu2024single}.
However, these theoretical algorithms are not always practical to implement.
While they achieve linear runtime in the blocklength, the constant prefactors, which may be exponential in the code's check weight, could be significant.
Alternatively, heuristic algorithms have been developed that are efficient and accurate in practice~\cite{Panteleev2021degeneratequantum,Roffe2020BPOSD,delfosse2021unionfinddecoderquantumldpc}.

One promising direction for developing decoding algorithms for quantum codes is to adapt successful decoders for classical codes to the quantum setting.
For example, the flip algorithm for expander codes~\cite{SipserSpielman96} was modified to small-set-flip~\cite{leverrier2015quantum} and its variants~\cite{Gu2023decoder,Leverrier2024Decoder,Dinur2023goodCode}, and belief propagation (BP) for classical LDPC codes~\cite{Pearl82BP,KFL01BP} was enhanced with ordered statistics decoding (OSD) post-processing (BP+OSD)~\cite{Panteleev2021degeneratequantum,Roffe2020BPOSD}.

Linear programming (LP) is another widely used decoder for classical LDPC codes~\cite{FeldmanLP,FeldmanThesis}. LP decoding is based on formulating the problem of finding the minimum-weight error consistent with a given syndrome as an integer program (IP) and relaxing it to a continuous optimization problem.
Thus, highly optimized LP solvers~\cite{gurobi} can be used to enable fast LP decoding.
LP decoders have provable performance guarantees for some classical LDPC codes~\cite{FeldmanConstantFraction,FeldmanCapacity}, such as expander codes, and provide certificates confirming minimum-weight solutions.
Although LP decoders remain relatively unexplored in the quantum setting, prior works have considered their direct extension, demonstrating promising performance for certain codes~\cite{FawziLP2021, LiVontobel, groues2022decoding}.

However, a significant challenge in applying LP decoders to quantum codes is the occurrence of \emph{fractional solutions} that do not correspond to valid physical corrections.
Such behavior was shown to be problematic for the surface code, codes with poor soundness~\cite{FawziLP2021}, and hypergraph product (HGP) codes with carefully tuned parameters~\cite{groues2022decoding}.
In this work, we ask:
\begin{center}
	\emph{Do fractional solutions arise\\
		generically in quantum LDPC codes?}
\end{center}

We answer this in the affirmative. In Sec.~\ref{sec:uncorrectableLPexamples}, we show that quantum codes generally allow for small, often constant-sized, error patterns which are uncorrectable by LP.
Such patterns, related to cycles in the code's Tanner graph, can prevent LP from attaining a threshold.

One approach for handling a fractional solution produced by the LP decoder is to round each entry of the solution independently~\cite{FawziLP2021, groues2022decoding}. While this method guarantees an integral solution, it is still not able to correct the problematic errors we identify.
Furthermore, the correction after independent rounding may not even have the same syndrome as the error, which our simulations show is the case in the vast majority of LP failures (Fig.~\ref{fig:wrongsyndromefraction}).
Motivated by this observation, we ask:
\begin{center} \emph{Are there more sophisticated rounding techniques that\\can improve the performance of LP decoders?} \end{center}

To explore this question, we draw inspiration from BP. BP is a fast message-passing decoder~\cite{KFL01BP} that is also known to perform well for classical codes~\cite{MacKayNeal95,mackay1996near,MacKay99} but has convergence issues for quantum codes due to short cycles in the Tanner graph~\cite{PoulinChung2008,BBANH15fifteenyears,Raveendran2021trappingsetsof}.
A common approach to address the non-convergence of BP is to apply OSD as a post-processing step.  In BP+OSD, the less reliable bits identified by BP are erased and Gaussian elimination is then used to recover a valid correction consistent with the given syndrome~\cite{Panteleev2021degeneratequantum, Roffe2020BPOSD}.
By interpreting the fractional outputs of LP as the reliability of the solution, we propose using OSD as a post-processing step in LP decoding (LP+OSD) to round the output to an integral solution satisfying the syndrome.
We show that the problematic error patterns identified in our analysis in Sec.~\ref{sec:uncorrectableLPexamples} can be corrected by OSD post-processing.

Another motivation for combining LP with OSD is that BP+OSD is already bottlenecked by the $O(n^3)$ runtime of the matrix inversion step in OSD. 
Therefore, replacing BP with LP, which also has complexity $O(n^3)$, does not increase the overall asymptotic complexity.
Since standalone LP decoding performs as well as—or better than—BP for many classical~\cite{FeldmanLP} and quantum codes~\cite{FawziLP2021, groues2022decoding}, one might expect that LP+OSD achieves lower logical error rates than BP+OSD. 
In this work, we find that this is the case for some mid-sized LDPC codes of up to 400 physical qubits (Figs.~\ref{fig:mainresults} and~\ref{fig:decodingcomparison}).

The remainder of the paper is structured as follows. After introducing some background on quantum LDPC codes and the LP and OSD decoders in Sec.~\ref{sec:prelim}, we describe the uncorrectable error patterns for LP in Sec.~\ref{sec:uncorrectableLPexamples}. In Sec.~\ref{sec:LPOSD}, we introduce our proposed LP+OSD decoder and explain why OSD post-processing helps in handling these problematic error patterns. We present the results of our numerical simulations in Sec.~\ref{sec:numerics}. Finally, we conclude in Sec.~\ref{sec:discussion} and discuss possible future directions.

\section{Preliminaries}
\label{sec:prelim}
We first review some facts about quantum LDPC codes and give an overview of the LP and OSD decoders.

\subsection{Classical and quantum codes}
A classical linear code is a subspace $C \subseteq \bb F_2^A$, where $A$ is a set indexing the bits. The classical code is often defined using a parity check matrix $H\in \bb F_2^{B\times A}$ with the code defined as the kernel: $C = \ker H$. Here, the rows $B$ index the set of parity checks that must be satisfied for a vector $x\in \bb F_2^A$ to be in the codespace. To visualize the code, we can represent it using a Tanner graph $T(C)$, which is a bipartite graph with vertices $A\sqcup B$. The edges are defined to connect $a\in A$ and $b\in B$ if the matrix element $H_{ba} = 1$.

A quantum stabilizer code $\mc C$ is a subspace of the $n$-qubit Hilbert space $\left(\bb C^2\right)^{\otimes n}$, defined  as the simultaneous $+1$-eigenspace of an abelian subgroup $S$ of the $n$-qubit Pauli group not containing $-I$, called the stabilizer group. When $S$ can be generated by operators from two sets $S_X$ and $S_Z$ consisting of Pauli operators of $X$ and $Z$ type, respectively, we call $\mc C$ a Calderbank-Shor-Steane (CSS) code~\cite{cssCalderbankShor,cssSteane}. We often identify $X$- or $Z$-type Pauli operators, their supports as subsets, and the corresponding vectors in $\bb F_2^n$ with each other. Associated with any CSS code are two parity check matrices $H_X\in \bb F_2^{m_X\times n}$ and $H_Z\in \bb F_2^{m_Z\times n}$, where the rows represent (supports of) the stabilizer generators of each type. When $H_X$ and $H_Z$ have at most a constant number of ones in each row and column, we say that the code is LDPC.

The number of logical qubits encoded in a CSS code is the dimension $k = n - \rk H_X - \rk H_Z$, and the rate is $k/n$. The distance is the weight of the smallest nontrivial logical operator, which is given by $d=\min\{d_X, d_Z\}$ with $d_X = \min \{|x|: x\in \ker H_Z\setminus \row{H_X}\}$ and $d_Z = {\min \{|x|: x\in \ker H_X\setminus \row{H_Z}\}}$. Here, $\row A$ is the rowspace of a matrix $A$ and $|x|$ is the Hamming weight of a binary vector.

We can also associate a Tanner graph $G=T(\mc C)=(V, E)$ with a CSS code $\mc C$. The vertices of $G$ split into three sets $V = C_Z \sqcup Q \sqcup C_X$, where $Q$ consists of the $n$ qubits and $C_Z$ and $C_X$ consist of the $m_Z$ and $m_X$ stabilizer generators of the corresponding type. The edges of $G$ are $\{i, c\}$ where $i\in Q$, $c\in C_Z\sqcup C_X$, and the qubit $i$ is in the support of the check $c$. Thus, $G$ has constant degree when the code is LDPC. We will sometimes work with the subgraphs $G_Z$ and $G_X$, which are the bipartite graphs generated by $C_Z \sqcup Q$ and $Q \sqcup C_X$, respectively.

We will use the convention that qubits in $Q$ are labeled $i$, checks in $C_X$ are labeled $j$, and checks in $C_Z$ are labeled $k$. The supports of $X$ and $Z$ stabilizers will be denoted $f$ and $g$, respectively, with generators labeled $f_j$ and $g_k$. We use the notation $N_X(\cdot)$ to denote the neighborhood in $G_X$. For example, $N_X(i)$ is a subset of $X$ checks.

\subsection{Code constructions}
\label{subsec:codeconstructions}
In this work, we will consider several families of quantum LDPC codes, including the surface code~\cite{Kitaev03}, HGP codes~\cite{Tillich2014LDPC}, lifted/balanced product codes~\cite{PKalmostlinear,Breuckmann2021LDPC}, and bivariate bicycle (BB) codes~\cite{bravyi2024high}.

An HGP code~\cite{Tillich2014LDPC} $\mc C$ is defined by two classical codes $C$, $C'$ with parity-check matrices $H\in \bb F_2^{B\times A}$, $H'\in \bb F_2^{B'\times A'}$, respectively. The qubits of $\mc C$ are $A\times A'\sqcup B\times B'$, the $X$ checks are associated with $A\times B'$, and the $Z$ checks are associated with $B\times A'$. The parity check matrices of $\mc C$ are $H_X = [\begin{array}{c|c} \iden_{A}\otimes H' & H^{\top}\otimes \iden_{B'}\end{array}]$ and $H_Z = [\begin{array}{c|c}H\otimes \iden_{A'} & \iden_{B}\otimes H'^{\top}\end{array}]$. Alternatively, whenever $\{a, b\}$ and $\{a', b'\}$ are edges in the Tanner graphs of the respective classical codes, then 
$\{(a, a'), (a, b')\}$, $\{(b, b'), (a, b')\}$ are edges in $G_X$ and $\{(b, a'), (a, a')\}$, $\{(b, a'), (b, b')\}$ are edges in $G_Z$.

The surface code~\cite{Kitaev03} can be thought of as the hypergraph product of two classical repetition codes. We will also consider the more qubit-efficient, rotated version of the surface code.

A lifted/balanced product code~\cite{PKalmostlinear,Breuckmann2021LDPC} $\overline{\mc C}$ is a quotient of an HGP code $\mc C$. More precisely, let $\mc C$ be defined by classical codes $C$ and $C'$. Suppose $C$ has a right group action for a group $G$; that is, $A$ and $B$ both have right $G$-actions and $\{ag, bg\}$ is an edge in $T(C)$ whenever $\{a, b\}$ is an edge in $T(C)$. Similarly, suppose $C'$ has a left $G$-action. Then $\overline{\mc C}$ is defined by identifying the qubit or check $(xg, y)$ with $(x, gy)$ for all $g\in G$, where $x\in A\sqcup B$ and $y\in A'\sqcup B'$. Thus, the number of qubits, $X$ checks, and $Z$ checks are each reduced by a factor of $|G|$ if the action is free.

A class of lifted product codes are BB codes~\cite{bravyi2024high}, which are defined by a quotient over the abelian group $\bb Z_\ell\times \bb Z_m$. They can alternatively be described by parity-check matrices $H_X = [\begin{array}{c|c}h & h'\end{array}]$, $H_Z = [\begin{array}{c|c}h'^\top & h^\top\end{array}]$, where $h, h'\in \bb F_2[x, y]/\langle x^\ell - 1, x^m - 1\rangle$ are bivariate polynomials which are typically the sums of three monomials. Concretely, $x = S_\ell\otimes I_m$ and $y = I_\ell \otimes S_m$, where $S_r\in \bb F_2^{r\times r}$ denotes a cyclic shift matrix.

\subsection{LP decoder}
\label{subsec:LPdecoder}
In the decoding problem, some unknown (Pauli) error $E$ is applied to the code state. By measuring stabilizers, we obtain the syndrome, which is a classical bit string $s$ telling us which generators anticommute with $E$. A decoder is a classical algorithm that takes $s$ as input and outputs a correction $\tilde E$. The algorithm succeeds if applying $\tilde E$ to the corrupted state returns the original code state, i.e., if $\tilde E E\in S$.

For CSS codes, we can consider decoding $X$ and $Z$ errors separately---although this may not be optimal as it ignores correlations due to $Y$ errors.
For $Z$ errors, we identify $E$ with its support as a vector $e\in \bb F_2^n$, and the relevant syndrome will be obtained from the $X$ stabilizer generators as $s=H_X e\in \bb F_2^{m_X}$.
Using $s$, the decoder should output $\tilde e$ such that $e+\tilde e\in \row{H_Z}$.
As the $X$ errors are decoded similarly, we will only consider $Z$ errors.

For independent and identically distributed $Z$ errors, the most likely error giving a syndrome $s$ is the one of lowest Hamming weight. This can be expressed as the solution of the (nonlinear) IP
\begin{align}\label{eq:IntegerProgramming}
	\begin{split}
		\min_{x\in\{0,1\}^{n}}\quad  &\sum_{i\in  Q} x_i\\
		\text{s.t.}\quad  &H_X x=s\mod 2\, .
	\end{split}
\end{align}
As IPs are generally intractable~\cite{Karp72}, a \emph{linear programming (LP) relaxation} can instead be considered. An LP is an optimization problem with real-valued variables such that the objective function is linear and all constraints are linear inequalities. Such problems can be solved efficiently~\cite{Schrijver11}. 
Ref.~\cite{FawziLP2021} considered the following LP to solve the decoding problem:
\begin{alignat}{2}
	\min \quad &\sum_{i\in Q} x_i \label{eq:LPRelaxationPrimal1}\\
	\text{s.t.} \quad & \sum_{S\in E_j^{s_j}}w_{j,S}=1 \quad && \forall j\in C_X\, ,\nn\\
	& \sum_{\substack{S\in E_j^{s_j} \\ S\ni i}}w_{j,S}=x_i \quad && \forall i\in Q,\ j\in C_X,\ \{i, j\}\in E\, ,\nn\\
	& x_i \ge 0 \quad &&\forall i\in  Q\, ,\nn\\
	& w_{j, S} \ge 0 \quad && \forall j\in  C_X,\ S\in E_j^{s_j}\, .\nn
\end{alignat}
Here, for any $j\in C_X$, we define
\begin{equation}
	E_j^{s_j} = \{S\subseteq f_j: |S| = s_j \mod 2\}
\end{equation} as the subsets with odd (when $s_j=1$) or even (when $s_j=0$) cardinality (including the empty set $\emptyset\in E_j^{0}$). In addition to $x_i$, the optimization problem includes auxiliary variables $w_{j,S}$ for $j\in C_X$ and $S\in E_j^{s_j}$. When $w_{j,S}\in\{0,1\}$, it corresponds to a binary decision variable indicating the subset of $f_j$ that is believed to be corrupted, i.e.,
\begin{equation}\label{eq:wjsdefinition}
	w_{j,S} = \begin{cases}
		1 & \res{x}{f_j} = S \, , \\
		0 & \text{otherwise}\, .
	\end{cases}
\end{equation}

We can now see why the optimization~\eqref{eq:LPRelaxationPrimal1} is a linear relaxation of the original IP~\eqref{eq:IntegerProgramming}. Suppose the last two constraints of the LP~\eqref{eq:LPRelaxationPrimal1} are replaced by their binary versions $x_i,w_{j,S}\in\{0,1\}$. Then the second constraint relates the auxiliary variables to the original variables $x_i$ according to Eq.~\eqref{eq:wjsdefinition}, and the first constraint enforces the parity of the error at each check. Hence, the $x_i$ variables of the optimal integral solution are the same as in the IP.

Allowing $x_i, w_{j,S}$ to take real values in $[0,1]$ yields an LP with number of variables and constraints scaling linearly with $n$. If the optimal solution $(x^*, w^*)$ obtained from the LP is integral, then $x^*$ is also the optimal solution of the IP. Such a solution would provide a certificate that $\tilde e = x^*$ is most likely error causing the syndrome. However, because the set of feasible solutions is larger for the LP, a fractional solution may yield a smaller objective value than for the IP. In this case, we may either declare a failure or output $\tilde e$ obtained by independently rounding each bit of $x^*$~\cite{FawziLP2021}.

\subsection{Ordered statistics decoding}
\label{subsec:OSD}
Ordered statistics decoding (OSD), first introduced in the context of classical codes~\cite{Fossorier1995OSD, Fossorier2001OSD}, is a post-processing technique that leverages soft information from decoders like BP (or for us, LP) to enhance decoding accuracy.
In quantum error correction, OSD post-processing on a potentially non-convergent BP output ensures a recovery operator that always maps the corrupted state back to the codespace~\cite{Panteleev2021degeneratequantum, Roffe2020BPOSD}. This approach has been shown to improve the performance of standard BP on quantum LDPC codes by several orders of magnitude~\cite{Panteleev2021degeneratequantum}.

Let $x\in [0,1]^n$ be a vector representing probabilities of $Z$ errors, which may be obtained from another decoder. The idea behind OSD is to commit the bits that are likely not erroneous, erase the rest, and use Gaussian elimination to find a solution consistent with the syndrome, which hopefully has low weight.

More precisely, we sort the qubits by values of $x_i$ in decreasing order and let $\mc S$ be the first $m_X = \rk(H_X)$ qubits that correspond to linearly independent columns of $H_X$. (By removing rows from $H_X$, we may assume without loss of generality that $H_X$ is full rank.) Let $\mc T$ be the rest of the qubits. In zeroth-order OSD (OSD-0), the correction $\tilde e$ is defined by $\tilde e_{\mc T} = 0$ and $\tilde e_{\mc S} = H_{\mc S}^{-1}s$. In higher-order OSD methods, we try solutions where $\tilde e_{\mc T}$ has a few bits set to one and output the solution where $(\tilde e_{\mc S}, \tilde e_{\mc T})$ has the lowest overall Hamming weight.
For example, in the combination sweep strategy (OSD-CS)~\cite{Roffe2020BPOSD}, all weight-one configurations of $\tilde e_{\mc T}$ as well as weight-two configurations limited to the first $\lambda$ bits are considered, where $\lambda$ is a parameter to be chosen.
The remaining bits $\tilde e_{\mc S}$ are chosen so that the syndrome is satisfied, i.e., $\tilde e_{\mc S} = H_{\mc S}^{-1}s + H_{\mc S}^{-1}H_{\mc T}\tilde e_{\mc T}$. Higher-order OSD generally leads to improved accuracy compared to OSD-0 at the cost of increased runtime. In OSD-CS, the parameter $\lambda$ is often set to a fixed value, such as $60$, to balance the two considerations.

\section{Constant-weight uncorrectable errors using LP decoding}
\label{sec:uncorrectableLPexamples}
We describe error patterns of constant weight in various quantum LDPC codes that result in fractional solutions to the LP relaxation~\eqref{eq:LPRelaxationPrimal1}.
Such error patterns cannot be corrected by the LP decoder without rounding, and we show that independent rounding of the fractional points often does not lead to a valid correction either.
This means that the LP decoder would not give a threshold for these codes, as there would be a constant probability of an uncorrectable error regardless of code size. To explain these error patterns, we use an equivalent ``error-based'' formulation of the LP~\cite{LiVontobel} and consider its dual optimization problem. The error-based LP is the one typically studied in the classical setting~\cite{FeldmanLP}, as we have access to the individual corrupted bits. We consider this formulation to adapt techniques that have been developed for analyzing classical codes.

\subsection{Error-based LP}
To use the error-based method of formulating the LP, we first find any error $e'$ satisfying the syndrome $s$, i.e., $H_X e' = s$. Then we solve the following LP:
\begin{alignat}{2}
\min \quad &\sum_{i \notin e'}x_i - \sum_{i\in e'} x_i 
\label{eq:errorbasedLP}\\
\text{s.t.} \quad &\sum_{S\in E_j^0}w_{j,S}=1 \quad &&\forall j\in C_X\, ,\nn\\
& \sum_{\substack{S\in E_j^0 \\ S\ni i}}w_{j,S}=x_i \quad &&\forall i\in Q,\ j\in C_X,\ \{i, j\}\in E\, ,\nn\\
& x_i\geq 0 \quad &&\forall i\in  Q\, ,\nn\\
& w_{j,S}\geq 0 \quad &&\forall j\in  C_X,\ S\in E_j^0\, .\nn
\end{alignat}
In the objective function, we identified $e'$ as the subset of indices $i$ with $e'_i=1$. When restricting to integral variables $x_i$, $w_{j,S}$, the feasible solutions for $x$ are the ones satisfying all the parity checks of the code. Therefore, the optimal integral solution is the closest codeword to $e'$, and the lowest-weight error for the syndrome $s$ is ${e' - x^* \pmod 2}$. 
A decoder using the error-based LP would output the solution $\tilde e$ with $\tilde e_i = x^*_i$ when $e'_i = 0$ and $\tilde e_i = 1 - x^*_i$ when $e'_i = 1$. These formulas are also valid when $x^*$ is a fractional solution, but $\tilde e$ would need to be rounded.

The error-based LP relaxation is equivalent to the syndrome-based one~\cite{FawziLP2021,groues2022decoding}. This equivalence is formally stated in the following lemma, with a proof provided in Appendix~\ref{app:reformulationLP} for completeness.

\begin{lemma}
	\label{lem:lpequivalnet}
	Let $e\in \bb F_2^n$ be a $Z$ error with syndrome $s=H_Xe$ and $e'$ be any error with the same syndrome. Then the LP~\eqref{eq:errorbasedLP}, defined with $e'$, is equivalent to the LP~\eqref{eq:LPRelaxationPrimal1}. More precisely, there is a bijection between feasible solutions given by the reflections
	\begin{align}
		x_i &= \begin{cases}x_i', \quad &i\notin e'\\1-x_i', \quad &i\in e'\end{cases}\quad \forall i\in Q\, ,\\
		w_{j,S} &= w_{j,S + U_j}'\quad \forall j\in C_X,\ S\in E_j^{s_j}\, ,
	\end{align}
	where $x_i$,$w_{j,S}$ are variables of the LP~\eqref{eq:LPRelaxationPrimal1}, $x_i'$,$w_{j,S}'$ are variables of the LP~\eqref{eq:errorbasedLP}, and $U_j = \res{e'}{f_j}$. Two solutions related by the bijection have the same objective values up to an additive constant $|e'|$. Furthermore, the LP decoders based on these solutions give the same output.
\end{lemma}

Associated with any LP is a dual LP. For the error-based LP~\eqref{eq:errorbasedLP}, the dual LP has variables $\sigma_j$ for $j\in C_X$ and $\tau_{ij}$ for $i\in Q$ and $j\in C_X$ such that $\{i,j\}\in E$. It is given as follows:
\begin{alignat}{2}
	\max \quad &\sum_{j\in C_X}\sigma_j \label{eq:dualLP}\\
	\text{s.t.} \quad &\sum_{j\in N_X(i)} \tau_{ij}\leq \gamma_i \quad &&\forall i\in Q\, ,\nn\\
	& \sum_{i\in S} \tau_{ij}\geq \sigma_j \quad &&\forall j\in C_X,\ S\in E^{0}_j\, ,\nn
\end{alignat}
where $\gamma_i = (-1)^{e'_i}$.

A standard result from optimization states that the primal and dual LPs have the same objective value~\cite{Schrijver11}.

\begin{lemma}
	\label{lem:duality}
	Let $x^*=(x^*_1,\dots,x^*_n)$ and $\sigma^*=(\sigma^*_1,\dots,\sigma^*_{m_X})$ be the optimal solutions of the LPs \eqref{eq:errorbasedLP} and \eqref{eq:dualLP}, respectively. Then
	\begin{align}
		\sum_{i \notin e'}x^*_i - \sum_{i\in e'} x^*_i = \sum_{j\in C_X}\sigma^*_j\, .
	\end{align}
\end{lemma}

\subsection{Problematic error configurations}\label{subsec:problematicPatterns}
Let $e$ be a reduced error, that is, one with the lowest Hamming weight out of the equivalent errors obtained by adding stabilizers. Because the error-based LP~\eqref{eq:LPRelaxationPrimal1} is equivalent to the syndrome-based LP~\eqref{eq:errorbasedLP} regardless of the representative $e'$ chosen (Lemma~\ref{lem:lpequivalnet}), we may assume $e' = e$. The reduced error condition means that the optimal integral solution of the LP~\eqref{eq:errorbasedLP} is 0 (for example, with the solution $x=0$).

For the different error patterns we construct, we argue by contradiction to show that they are not correctable by LP. If the LP decoder succeeds without rounding, the optimal solution must not be fractional. Therefore, the optimal value of the LP~\eqref{eq:errorbasedLP} is 0. By Lemma~\ref{lem:duality}, so is the optimal value of the dual LP~\eqref{eq:dualLP}.
When this is the case, each $\sigma_j$ must be 0 because the second constraint with $S=\emptyset$ implies that the variables $\sigma_j$ are non-positive. When $\sigma_j=0$, the second constraint can alternatively be written as $\tau_{ij}+\tau_{i'j}\ge 0$ for all $j\in C_X$ and $i,i'\in f_j$ with $i\ne i'$. Therefore, we obtain a simplified set of inequalities.

\begin{proposition}
	If the LP~\eqref{eq:LPRelaxationPrimal1} outputs an integral solution for a reduced error $e$, there exist edge weights $\tau_{ij}$ that satisfy the inequalities
	\begin{alignat}{2}
		\sum_{j\in N_X(i)}\tau_{ij}&\le \gamma_i \quad &&\forall i\in Q \label{eq:dualLPsimplified1}\, ,\\
		\tau_{ij}+\tau_{i'j}&\ge 0 \quad &&\forall j\in C_X,\ i, i'\in f_j,\ i\ne i' 
		\label{eq:dualLPsimplified2}\, .
	\end{alignat}
\end{proposition}

If Inequality~\eqref{eq:dualLPsimplified1} is satisfied, it will also be satisfied if we decrease any $\tau_{ij}$. Considering Inequality~\eqref{eq:dualLPsimplified2}, if any $\tau_{ij}<0$, we may set $\tau_{i'j}=-\tau_{ij}$ for the other $i'\in f_j$. This leads us to an equivalent and intuitive way to formulate the above inequalities in terms of hyperflow of ``poison''~\cite{DaskalakisLP}. Each corrupted qubit $i\in e$ has one unit of poison that must be transported away. In order to transport some amount of poison $\beta_j=-\tau_{ij}$ through a neighboring check $j\in C_X$, it has to spread to all other neighbors of $j$ ($\tau_{i'j}=\beta_j$ for $i'\in f_j\setminus\{i\}$). Each uncorrupted qubit $i\in Q\setminus e$ can absorb one unit of poison. Inequalities~\eqref{eq:dualLPsimplified1} and~\eqref{eq:dualLPsimplified2} have a solution if and only if there is a valid hyperflow. Using this formulation, we prove a fact about errors that are half of a $Z$ stabilizer.

\begin{lemma}
	\label{lem:halfZstab}
	Let $g$ be a $Z$ stabilizer of even weight and $e$ be a reduced error such that $|e\cap g|=|g|/2$. If the LP~\eqref{eq:errorbasedLP} outputs an integral solution, then any edge weights satisfying Inequalities~\eqref{eq:dualLPsimplified1} and~\eqref{eq:dualLPsimplified2} must saturate the first inequality for all $i\in g$.
\end{lemma}
\begin{proof}
	Because any $j\in C_X$ commutes with $g$, the overlap $|f_j\cap g|$ must be even. In particular, it must be at least $2$ for any $j\in N_X(g)$. This means whenever we try to move poison away from a qubit in $g$, it must be distributed somewhere else in $g$ (and possibly to other qubits). Since we need to move $|e\cap g|=|g|/2$ units of poison away from the corrupted vertices in $g$ and there are only $|g\setminus e|=|g|/2$ uncorrupted vertices to absorb the poison in $g$, the first inequality~\eqref{eq:dualLPsimplified1} for all $i\in g$ must be saturated.  
\end{proof}

We are now ready to present several configurations of small uncorrectable errors which can occur in codes with stabilizers satisfying certain conditions. 

\begin{proposition}
	\label{counter:Zoverlap}
	Let $g,g'$ be two $Z$ stabilizers of even weight such that $|g\cap g'|\ge 2$.
	Furthermore, suppose all $X$ checks $j\in N_X(g\cap g')$ satisfy $f_j\cap (g+g')\ne \emptyset$.
    Then the LP decoder~\eqref{eq:LPRelaxationPrimal1} cannot correct a reduced error $e$ defined with $|\res{e}{g\setminus g'}| = \lfloor |g\setminus g'|/2 \rfloor$, $|\res{e}{g'\setminus g}| = \lceil |g'\setminus g|/2 \rceil$, and $|\res{e}{g\cap g'}| = 1$.
\end{proposition}

\begin{proof}
	Assume for a contradiction that $e$ can be corrected by the LP. Because $|e\cap (g+g')|=|g+g'|/2$, Lemma~\ref{lem:halfZstab} implies Inequality~\eqref{eq:dualLPsimplified1} is saturated for all $i\in g+g'$. But to remove the one unit of poison in $g\cap g'$, it must go through an $X$ check $j\in N_X(g\cap g')$. This check also has support on $g+g'$ and will therefore spread extra poison to the set, rendering Inequality~\eqref{eq:dualLPsimplified1} impossible for some $i\in g+g'$.
\end{proof}

The above proposition states that there will typically be an uncorrectable error pattern if two $Z$ stabilizers overlap on at least two qubits. By commutation with $Z$ stabilizers, $f_j$ must have at least weight-two support on $g$ and $g'$. The $X$-check condition expresses that its support within $g\cup g'$ cannot be limited to $g\cap g'$. The weight of the error constructed cannot be reduced by adding $g$, $g'$, or $g+g'$. Given the many ways to choose $e$ satisfying the conditions, it is likely that one choice will be reduced and thus uncorrectable by LP.

We can explicitly demonstrate a fractional solution of the LP~\eqref{eq:LPRelaxationPrimal1} with objective value smaller than $|e|$ that causes the LP decoder to fail. Set
\begin{equation}
	\label{eq:Zoverlapprimalsol}
	x_i = \begin{cases}
		1/2, & i\in g+g',\\
		0, & \text{otherwise}\, ,
	\end{cases}
\end{equation}
which gives objective value $|g+g'|/2 = |e| - 1$. To define the $w_{j,S}$ variables, note that for any check $j$, the overlap $S_j \coloneqq f_j\cap (g+g')$ has even cardinality due to commutativity of the checks. Suppose $|S_j|\ge 2$. If $s_j=0$, we set $w_{j,S_j} = w_{j,\emptyset}= 1/2$. If $s_j=1$, let $i\in S_j$, and we set $w_{j,\{i\}} = w_{j,S_j\setminus \{i\}} = 1/2$. When $S_j = \emptyset$, the set $f_j$ cannot have support on $g\cap g'$ either by the conditions of the proposition, so $s_j=0$, and we define $w_{j,\emptyset} = 1$. All other $w_{j,S}$ are set to 0.

Fig.~\ref{subfig:Zoverlap} illustrates an example of the error pattern in Proposition~\ref{counter:Zoverlap}. The stabilizers $g=Z_1Z_2Z_3Z_4Z_5Z_6$ and $g'=Z_5Z_6Z_7Z_8Z_9Z_{10}$ overlap on the two qubits 5 and 6. The $Z$ error shown affects half of $g\setminus g'$ (qubits 1 and 3), half of $g'\setminus g$ (qubits 7 and 8), and one qubit in $g\cap g'$ (qubit 5).

\begin{figure}[htpb]
	\centering
	\subfloat[]{\includegraphics[width=0.16\textwidth]{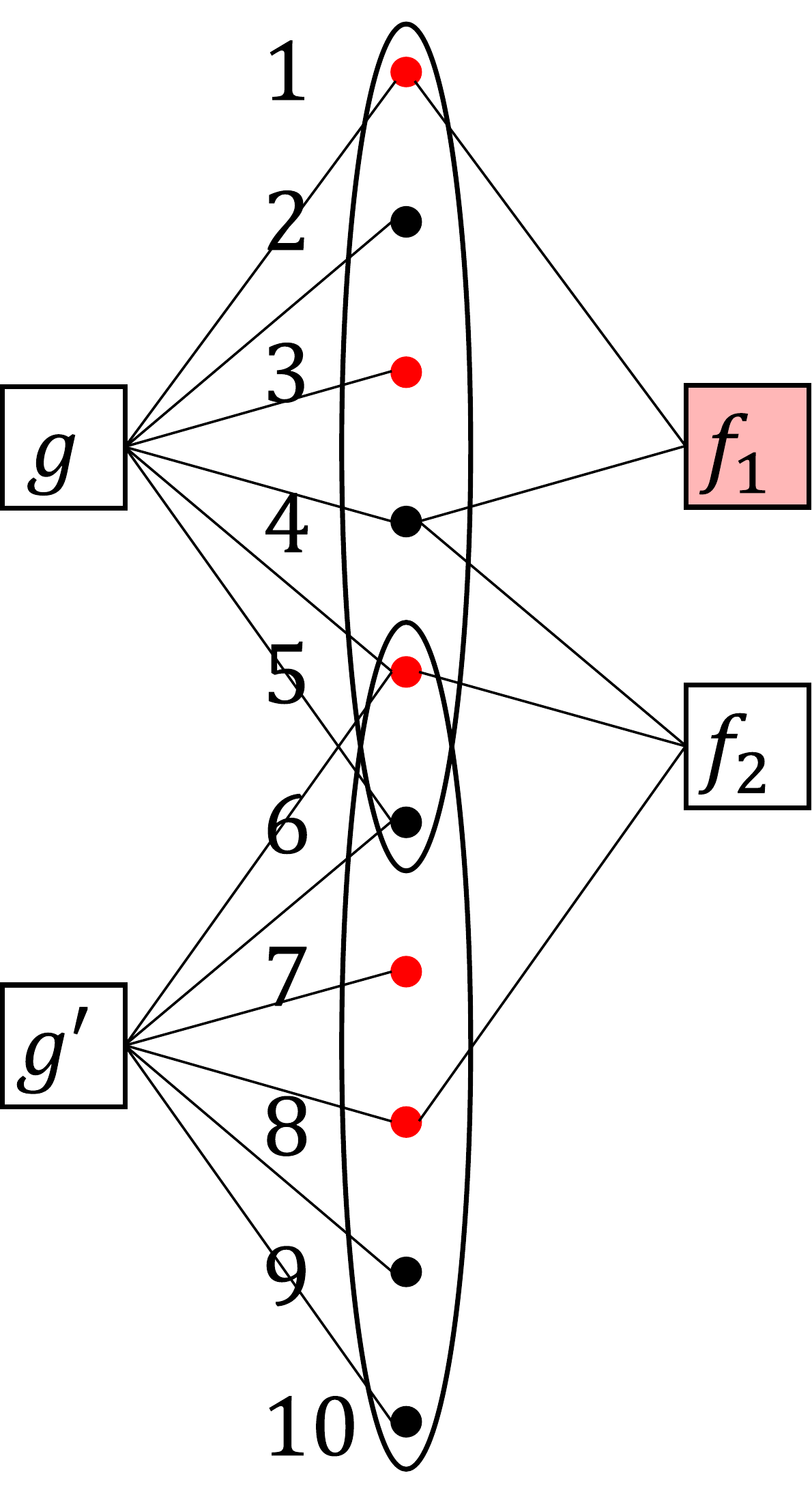}%
		\label{subfig:Zoverlap}}
	\hfil
	\subfloat[]{\includegraphics[width=0.3\textwidth]{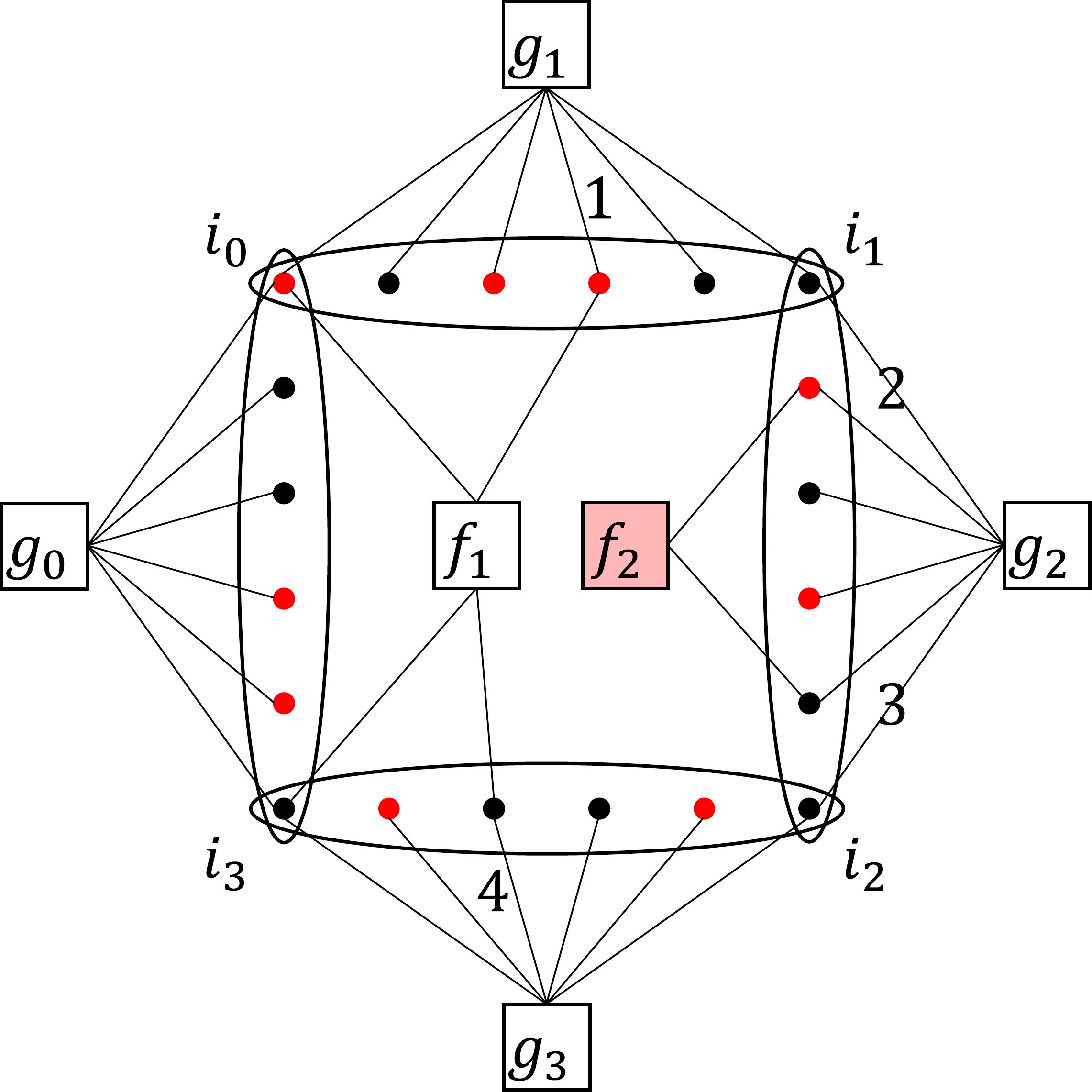}%
		\label{subfig:Zcycle}}
	\caption{Error patterns that are uncorrectable by LP~\eqref{eq:LPRelaxationPrimal1}. In the parts of the Tanner graphs illustrated, dots represent qubits with $Z$ errors affecting those in red, and squares represent checks with the violated $X$ checks highlighted. (a) Proposition~\ref{counter:Zoverlap}. An objective value of 4 is achieved with $x_i=1/2$ for $i=1,2,3,4,7,8,9,10$ and $w_{f_1, \{1\}} = w_{f_1, \{4\}} = w_{f_2, \{4, 8\}} = w_{f_2, \emptyset} = 1/2$. (b) Proposition~\ref{counter:Zcycle}. An objective value of 8 is achieved with $x_i=1/2$ for $i\ne i_0, i_1, i_2, i_3$ and $w_{f_1, \{1, 4\}} = w_{f_1, \emptyset} = w_{f_2, \{2\}} = w_{f_2, \{3\}} = 1/2$.}
	\label{fig:counterexamples}
\end{figure}

More generally, we have uncorrectable errors associated with cycles in $Z$ Tanner graph $G_Z$.

\begin{proposition}
	\label{counter:Zcycle}
	Let $g_0,g_1,\dots,g_{K-1},g_K=g_0$ be $Z$ stabilizers of even weight with intersections $g_k\cap g_{k+1} = \{i_k\}$ for all $k$ and $g_k\cap g_{k'} = \emptyset$ for $|k-k'|\ge 2$.
	Suppose all $X$ checks $j\in N_X(I)$ satisfy $f_j\cap(\sum_{k=0}^{K-1} g_k) \ne \emptyset$, where $I = \{i_k\}_{k=0}^{K-1}$.
    Then the LP decoder~\eqref{eq:LPRelaxationPrimal1} cannot correct a reduced error $e$ defined with $\res{e}{I} = i_0$ and $|\res{e}{g_k}| = |g_k|/2 - 1$ for all $k$.
\end{proposition}

\begin{proof}
	Suppose for a contradiction that $e$ can be corrected by the LP. Because $|e\cap \sum_{k=0}^{K-1}g_k| = |\sum_{k=0}^{K-1}g_k|/2$, Lemma~\ref{lem:halfZstab} implies that Inequality~\eqref{eq:dualLPsimplified1} is saturated for all $i\in \sum_{k=0}^{K-1}g_k$. Now consider the poison on $i_0$. It must go through an $X$ check $j\in N_X(i_0)$. Because $f_j$ also has support on $\sum_{k=0}^{K-1}g_k$, it will spread too much poison to that set, and Inequality~\eqref{eq:dualLPsimplified1} will be violated for some qubit in $\sum_{k=0}^{K-1}g_k$.
\end{proof}

We remark that the even-weight stabilizers condition in the proposition is not restrictive because we can always add two odd-weight stabilizer generators to obtain an even-weight stabilizer.
Again, no stabilizer generated by the $\{g_k\}$ can reduce the weight of $e$ as defined, so $e$ will generically be reduced.

A fractional solution of the syndrome-based LP~\eqref{eq:LPRelaxationPrimal1} is obtained by setting
\begin{equation}
	\label{eq:Zcycleprimalsol}
	x_i = \begin{cases}
		1/2, & i\in \sum_{k=0}^{K-1}g_k,\\
		0, & \text{otherwise}\, .
	\end{cases}
\end{equation}
This gives objective value $\frac 1 2 |\sum_{k=0}^{K-1}g_k| = |e| - 1 < |e|$. Note that for any check $j$, the overlap $S_j \coloneqq f_j\cap \sum_{k=0}^{K-1}g_k$ has even cardinality. This allows us to define the auxiliary variables $w_{j,S}$ as follows. When $|S_j|\ge 2$ and $s_j = 0$, we set $w_{j,S_j} = w_{j,\emptyset}= 1/2$. When $|S_j|\ge 2$ and $s_j = 1$, we let $i\in S_j$ and set $w_{j,\{i\}} = w_{j,S_j\setminus \{i\}} = 1/2$. When $S_j = \emptyset$, we must have $s_j=0$, so we set $w_{j,\emptyset} = 1$. All other $w_{j,S}$ are set to 0.

See Fig.~\ref{subfig:Zcycle} for an example of the error pattern in Proposition~\ref{counter:Zcycle}. Here, the $Z$ stabilizers $g_0, g_1, g_2, g_3$ overlap on the qubits $i_0=g_0\cap g_1$, $i_1=g_1\cap g_2$, $i_2=g_2\cap g_3$, $i_3=g_3\cap g_0$. Since each stabilizer has weight 6, we corrupt $6/2-1=2$ of the qubits in each of the sets $g_k\setminus\{i_{k-1},i_k\}$ for $k=0,1,2,3$ as well as the qubit $i_0$.

Because the fractional solutions in Propositions~\ref{counter:Zoverlap} and~\ref{counter:Zcycle} have variables $x_i=1/2$, independent rounding will not give a valid solution most of the time, as it is unclear whether to round each coordinate to 0 or 1.

In Proposition~\ref{counter:Zcycle}, the $Z$ stabilizers give the cycle $0, i_0, 1, i_1, \dots, K-1, i_{K-1}, K=0$ in $G_Z$ when the $g_k$ are generators.
If the condition $g_k\cap g_{k'} = \emptyset$ for $|k-k'|\ge 2$ is not satisfied, we could obtain a smaller cycle in $G_Z$ that we can apply the proposition to.
The case where $|g_k\cap g_{k+1}|\ge 2$ is covered by Proposition~\ref{counter:Zoverlap} (which also gives the cycle $k_0, i_0, k_1, i_1, k_0$ in $G_Z$, where $k_0$ and $k_1$ label the $Z$ stabilizer generators $g$ and $g'$, respectively, and $i_0, i_1\in g\cap g'$). In summary, cycles in $G_Z$ typically give rise to error patterns that are not correctable by LP.

This observation implies that quantum LDPC codes can exhibit small, uncorrectable errors under LP decoding, with the size of such errors determined by the girth (i.e., the length of the shortest cycle) of the Tanner graph $G_Z$. How large can this girth be? A standard counting argument shows that the girth of $G_Z$ is at most $O(\log n)$. Because each stabilizer generator has constant size, the problematic error patterns arising from the shortest cycle of $G_Z$ will also have weight $O(\log n)$. We also note that for cycles of size greater than the stabilizer weights, the condition $f_j\cap \sum_{k=0}^{K-1} g_k\ne \emptyset$ for $j\in N_X(I)$ is automatically satisfied since otherwise, commutation with each of the $g_k$ would imply $f_j\supseteq I$. Thus, the LP decoder in general cannot correct adversarial errors of weight $\omega(\log n)$. Moreover, if it could correct stochastic errors, the suppression of errors below the threshold would be at most polynomial (instead of exponential) in the blocklength.

Some quantum LDPC codes do have logarithmically scaling girths in their $X$ and $Z$ Tanner graphs~\cite{pacentivasic2024quantummarguliscodes}. However, many code constructions exhibit \emph{constant} girth and would likely support constant-sized uncorrectable errors.
In the next section, we present examples of such codes.

\subsection{Uncorrectable errors in different code families}
\label{subsec:uncorrectableerrorsHGP}
Using Propositions~\ref{counter:Zoverlap} and~\ref{counter:Zcycle}, we show that many families of quantum LDPC codes exhibit problematic structures in their Tanner graphs, meaning that LP by itself does not have a threshold for these codes. To do this, we demonstrate constant-sized cycles in $G_Z$, find a corresponding reduced error $e$, and show that the condition on the $X$ checks is satisfied. In Appendix~\ref{app:quantumexpandercodes}, we present a rigorous argument in the case of quantum expander codes~\cite{leverrier2015quantum}, which are HGP codes constructed from two classical codes with expanding Tanner graphs. For other codes, we often cannot certify that error $e$ constructed is reduced. This is because even though combinations of $Z$ generators from the cycle cannot reduce the error weight, other $Z$ stabilizers might. However, given that many constant-sized cycles are present in the code and $e$ may be defined in many ways from a given cycle, it is unlikely that no such errors would be reduced.

Consider an HGP code $\mc C$ from two classical LDPC codes $C$, $C'$. For simplicity, assume that the generator weights of $\mc C$ are even. Otherwise, we can make the same argument with $Z$ stabilizers which are sums of two generators. We exhibit a constant-length cycle in the $Z$ Tanner graph by considering two cases.

First, suppose $T(C)$ and $T(C')$ have girths greater than six.
Let $P = b_1,a_1,b_2,a_2,b_3$ and $P' = a_1',b_1',a_2',b_2',a_3'$ be paths in the $T(C)$ and $T(C')$, respectively, where each $a_i\in A$, $a_i'\in A'$, $b_i\in B$, and $b_i'\in B'$. Consider the cycle $\Gamma$ in $G_Z$ obtained by concatenating $P\times a_1'$, $b_3\times P'$, $P\times a_3'$, and $b_1\times P'$, i.e., 
\begin{align}
	&(b_1,a_1'), (a_1,a_1'), (b_2,a_1'), (a_2,a_1'), (b_3,a_1'), (b_3,b_1'),\nn\\
	&(b_3,a_2'), (b_3,b_2'), (b_3,a_3'), (a_2,a_3'), (b_2,a_3'), (a_1,a_3'),\nn\\
	&(b_1,a_3'), (b_1,b_2'), (b_1,a_2'), (b_1,b_1'), (b_1,a_1')\, . \label{eq:HGPcycle}
\end{align}
We show that the intersections of $Z$ generators from $\Gamma$ are exactly the qubits in the cycle. For two generators to overlap, either the first or the second coordinate must be the same. Without loss of generality, suppose $(b_i, a'), (b_j, a')\in \Gamma$ have a common neighbor $(a, a')\notin \Gamma$. Then the subpath of $P$ from $b_i$ to $b_j$ concatenated with $b_j, a, b_i$ would give a cycle in $T(C)$ of length at most six, contradicting its girth.

Now, suppose one of the Tanner graphs of the classical codes, say $T(C)$, has girth at most six. Let $P = b_1, a_1, \dots, b_r = b_1$ be a minimal cycle, where each $a_i\in A$ and $b_i\in B$ and $r\le 3$. Then $P\times a'$ for any $a'\in A$ gives a cycle in $G_Z$. Intersections of two generators from this cycle cannot overlap on other qubits due to the minimality of the cycle $P$.

Let $I$ be the qubits in $\Gamma$. There is no $X$ check that contains $I$ since the qubits in $A\times A'$ do not all have the same first coordinate. Therefore, any $X$ check $j\in N_X(I)$ intersects the sum of the $Z$ checks in the cycle due to commutation with each $Z$ check. Aside from the reduced error condition, Proposition~\ref{counter:Zcycle}, or Proposition~\ref{counter:Zoverlap} if $r=2$ in the second case, gives a constant-weight error supported on the $Z$ generators of the cycle that cannot be corrected by the LP decoder. See Fig.~\ref{fig:HPCcounterexample} for an example.

\begin{figure}[ht]
	\centering
	\includegraphics[width=0.4\textwidth]{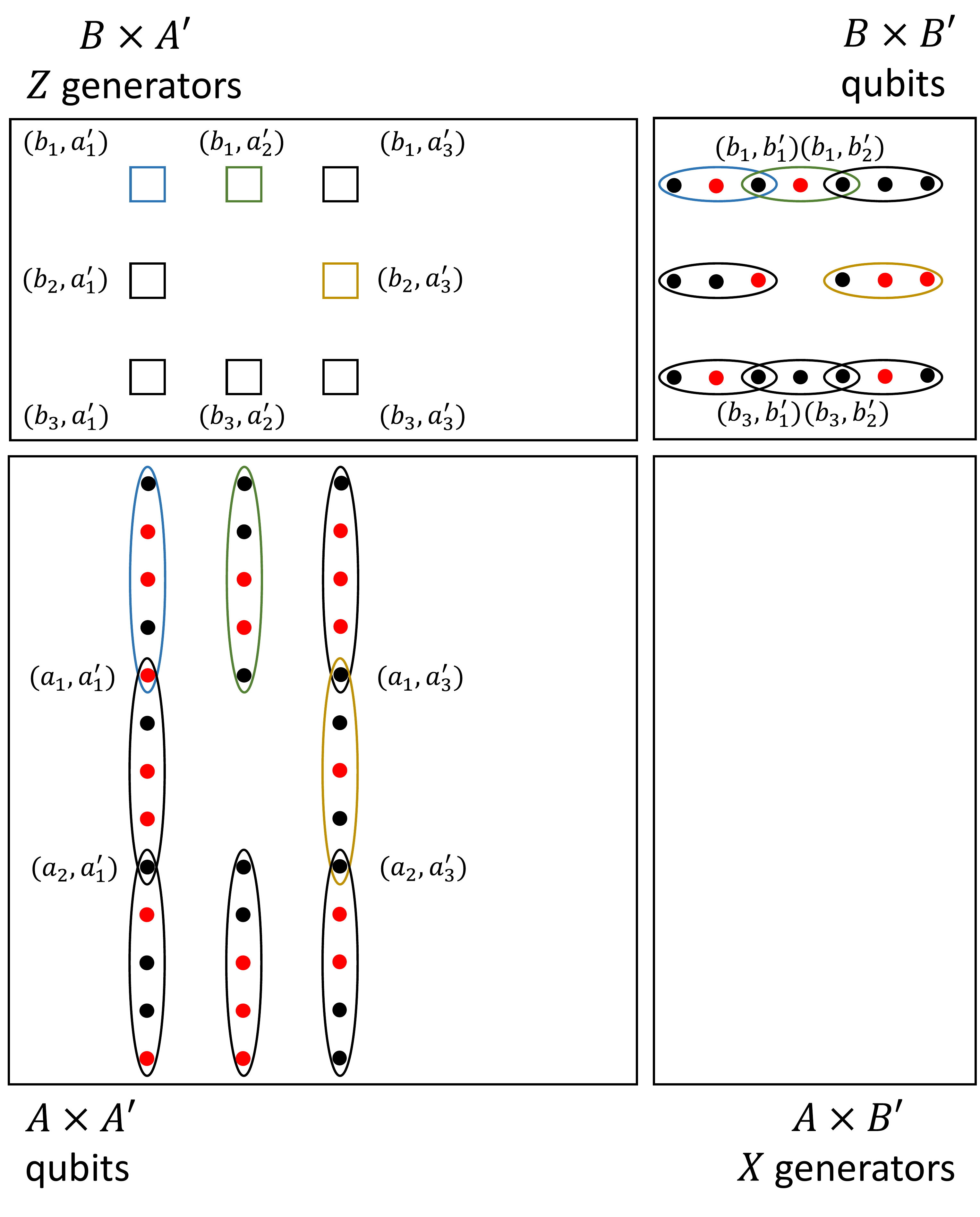}
	\caption{A constant weight error pattern in an HGP code that cannot be corrected using LP. The cycle~\eqref{eq:HGPcycle} is labeled, and some of the $Z$ generators and their qubit supports are shown with different matching colors. Corrupted qubits are colored in red.}
	\label{fig:HPCcounterexample}
\end{figure}

Because lifted product codes are quotients of HGP codes, they also typically contain constant-weight uncorrectable errors. Let $\overline{\mc C}$ be a lifted product code with check weights at most $w = O(1)$, which is obtained as the quotient of an HGP code $\mc C$ by a group $G$. Any cycle in the $Z$ Tanner graph of $\mc C$ projects to a cycle in the $Z$ Tanner graph of $\overline{\mc C}$. In particular, consider the cycle in $\mc C$ created from the two paths $P=b_1,a_1,\dots, b_r$ and $P'=a_1',b_1',\dots,a_r'$ by concatenating $P\times a_1'$, $b_r\times P'$, $P\times a_r'$, and $b_1\times P'$. If we assume that the projection does not identify two distinct points in the cycle and that the projected cycle in $T(\overline{\mc C})$ is minimal, we obtain an uncorrectable error pattern. By choosing $r=w/4+2$, the $X$-check condition is automatically satisfied as the cycle has length $4(r-1) > w$.

In particular, because the surface code is an HGP code and BB codes are lifted product codes, they also contain constant-weight error patterns that are not correctable by LP.

We remark that in a recent work, Grou\`es~\cite{groues2022decoding} has independently discovered a different type of constant-weight uncorrectable error pattern in certain HGP codes, which is also related to the presence of $Z$ stabilizers but has no immediate connection to cycles in the Tanner graph.

\section{LP+OSD decoder}\label{sec:LPOSD}
We now describe our proposed decoder based on LP and OSD post-processing to decode quantum LDPC codes. For a CSS code with parity check matrices $H_X$, $H_Z$, suppose we are given the syndrome $s$ of a $Z$ error. We first run the syndrome-based LP~\eqref{eq:LPRelaxationPrimal1} to obtain a solution $x$. If $x$ is integral, then it is guaranteed to be the minimum-weight error giving the syndrome $s$, so we output this solution. Otherwise, $x$ contains fractional components. To obtain an integral solution satisfying the syndrome $s$, we apply OSD post-processing, where we interpret the fractional components $x_i$ as the probability that qubit $i$ suffered a $Z$ error. Either zeroth- or higher-order OSD may be used, depending on the accuracy desired.

One thing to note is that a fractional solution from the LP often contains many entries which are zero. Thus, there is freedom in ordering the corresponding qubits for OSD post-processing, which may affect the rounded solution. In the examples of LP decoding failure found in Sec.~\ref{sec:uncorrectableLPexamples}, the LP assigns zero to one qubit in the error set, and it would be beneficial if this coordinate was erased by OSD. This qubit would typically be connected to a nontrivial syndrome in the Tanner graph. Therefore, we use the heuristic that qubits which are close to the nonzero syndromes are more likely to have an error than those which are farther away. When ordering the qubits according to probability of being flipped for OSD, we break ties by the distance to a nonzero syndrome. Pseudocode for our algorithm is given in Algorithm~\ref{alg:LPOSD}.

\begin{algorithm}[H]
	\caption{LP+OSD Decoder}
	\label{alg:LPOSD}
	\textbf{Input:} Parity-check matrix $H_X$, syndrome $s$ of a $Z$ error\\
	\textbf{Output:} Correction $\tilde e$
	\begin{algorithmic}[1]
		\State $x \gets$ solution to the LP~\eqref{eq:LPRelaxationPrimal1}
		\If{$x$ is integral} 
			\State \Return $x$ \Comment{minimum-weight error consistent with $s$}
		\Else
			\ForAll{$i\in Q$}
				\State $d_i \gets \operatorname{dist}_{G_X}(i, \{j\in C_X: s_j=1\})$
			\EndFor
			\State $Q'\gets$ qubits sorted by descending order of $x_i$, breaking ties by ascending order of $d_i$
			\State $\tilde e\gets \operatorname{OSD-0}(s)$ or $\operatorname{OSD-CS}(s)$ \Comment{$\mc S, \mc T$ defined by $Q'$}
			\State \Return $\tilde e$
		\EndIf
	\end{algorithmic}
\end{algorithm}

The decoder as described works in the code capacity setting. However, it can easily be adapted to handle phenomenological or circuit-level noise; see Appendix~\ref{app:circuitlevelnoise}.

The worst-case complexity of our LP+OSD decoder is $n^{\omega+o(1)}$ in the asymptotic limit, where $\omega\approx 2.37$ is the matrix multiplication exponent~\cite{multiplicationexponent}. Solving the LP requires $n^{\omega+o(1)}$ time for an $n$-qubit LDPC code~\cite{LPcomplexity}.
In the OSD subroutine, the decoder first ranks bits based on probability and, in case of ties, computes shortest distances via Dijkstra's algorithm. Thus, it takes time $O(n \log n)$ to determine $\mc S$ and $\mc T$~\cite{FredmanTarjan84Dijkstracomplexity}.
We then precompute $H_{\mc S}^{-1}$, which is done in time $n^{\omega+o(1)}$~\cite{Introtoalgs}.
If we run OSD-0, we compute $\tilde e_{\mc S} = H_{\mc S}^{-1}s$ in time $O(n^2)$.
If we run OSD-CS with a constant parameter $\lambda = O(1)$, we need to compute $\tilde e_{\mc S} = H_{\mc S}^{-1}s + H_{\mc S}^{-1}H_{\mc T}\tilde e_{\mc T}$ for all the $O(n)$ candidate values of $\tilde e_{\mc T}$.
This can be done in time $n^{\omega+o(1)}$ if we compute $H_{\mc S}^{-1}H_{\mc T}\tilde E_{\mc T}$, where $\tilde E_{\mc T}$ is a matrix with columns that are the different values of $\tilde e_{\mc T}$.
Picking the solution where $(\tilde e_{\mc S}, \tilde e_{\mc T})$ has the lowest Hamming weight takes time $O(n^2)$.
Therefore, the total running time is $n^{\omega+o(1)}$, although the optimized matrix multiplication algorithms have extremely large constant overheads.
For practical purposes, conventional $O(n^3)$ matrix multiplication and inversion would be used, and the LP would be solved using the simplex algorithm, which is often fast but has exponential complexity in the worst case.

\subsection{Problematic error patterns under LP+OSD}
\label{subsec:counterexamplesOSD}
Now we show how OSD post-processing can help in correcting the problematic error patterns in Sec.~\ref{sec:uncorrectableLPexamples}. First, let $e$ be defined as in Proposition~\ref{counter:Zoverlap} from constant-weight stabilizers $g$ and $g'$.
Consider $x$ from Eq.~\eqref{eq:Zoverlapprimalsol}, which is a fractional solution to the syndrome-based LP~\eqref{eq:LPRelaxationPrimal1}.
When OSD from Sec.~\ref{subsec:OSD} is applied to this solution, it will erase a set of qubits $\mc S$, starting from the elements of $g+g'$ because they are assigned higher probability of being corrupted by the LP. The fact that the columns of $H_X$ corresponding to $\mc S$ are linearly independent is equivalent to $\mc S$ not containing any $Z$ stabilizers or $Z$ logical operators. Assuming that the code has macroscopic distance and that there are no other stabilizer generators contained in $g\cup g'$, the only relevant stabilizers are $g$ and $g'$. Therefore, $\mc S$ contains all qubits of $g+g'$ except one, which we denote $i_1$. Without loss of generality, we can assume that $i_1\notin e$; otherwise, the same analysis would hold for the equivalent error $e+g+g'$ and we would show that OSD returns this error instead.

Consider the case when $\mc S$ contains the corrupted qubit $i_0\in g\cap g'$. Then $e\subseteq \mc S$, so $e$ is a valid solution for the syndrome that is consistent with the unerased qubits set to 0. Because there is a unique solution satisfying the syndrome such that $\tilde e_{\mc T}=0$, OSD-0 will return the actual error $\tilde e = e$.

If $\mc S$ does not contain the qubit $i_0$, OSD-0 might not find the solution $e$ because it will set $\tilde e_{i_0} = 0$ (although it may find a stabilizer-equivalent solution). However, when we search over weight-one configurations of $\tilde e_{\mc T}$ in OSD-CS, one of them will have $\tilde e_{i_0} = 1$ and $\res{\tilde e}{\mc T\setminus\{i_0\}} = 0$. In that iteration, the solution $e$ will be found since $\res{e}{\mc T} = \res{\tilde e}{\mc T}$.
There are no logically inequivalent errors with lower weight as $e$ has constant weight and the code has growing distance. Therefore, the final output of OSD-CS will be an error equivalent to $e$.

The error configuration of Proposition~\ref{counter:Zcycle} is similar. Let $e$ be the error defined from constant-weight stabilizers $g_0$, \dots, $g_{K-1}$, and consider $x$ defined by Eq.~\eqref{eq:Zcycleprimalsol}. Assuming the code has macroscopic distance and no other $Z$ stabilizers are contained in $\cup_{k=0}^{K-1}g_k$, the only relevant stabilizers are the $g_k$. The erased set $\mc S$ will contain all qubits in $\sum_{k=0}^{K-1}g_k$ except one, which we can assume is not an error without loss of generality by possibly replacing $e$ with the equivalent error $e+\sum_{k=0}^{K-1}g_k$. Using the same reasoning as before, the solution $e$ will be found by OSD-0 if $i_0\in \mc S$ or by OSD-CS if $i_0\notin \mc S$.

\subsection{Comparison with other LP-based decoders}
Before presenting our numerical results, we discuss the similarities and differences between our LP+OSD decoder and other LP-based decoders. For the surface code and certain other topological codes, minimum-weight perfect matching (MWPM) gives the most likely error with a given syndrome. Efficient implementations of MWPM are often based on the blossom algorithm~\cite{Edmonds1965blossom,Fowler2012decoder,Higgott2025sparseblossom,Wu2023fusionblossom}, which is also formulated in terms of an LP relaxation of the IP that describes the minimum-weight solution of the decoding problem. In that setting, because the algorithm makes use of the special structure of the decoding problem, i.e., each error causing at most two syndromes, the algorithm always outputs the optimal solution in polynomial time.

For general quantum LDPC codes, the hypergraph minimum-weight parity factor (MWPF) decoder has been proposed as a generalization of the MWPM decoders~\cite{Wu2025MWPF}. As such, MWPF is also based on solving an LP. The constraints for the LP in MWPF correspond to larger and larger subsets of the decoding graph that are ``invalid clusters'' given the observed syndrome. There is an exponential number of constraints, requiring a specialized algorithm that considers a polynomial number of relevant constraints at any given time. MWPF can often find the most likely error (with a certificate of having done so), but there is no guarantee of the runtime for a general code. In contrast, the LP relaxation we consider based on Refs.~\cite{FawziLP2021, LiVontobel, groues2022decoding}, contains more fine-grained information about the local neighborhoods of each check. Because there are only a linear number of total constraints, a generic LP solver may be used for the LP portion of the decoder, and it will always terminate in polynomial time. The downside is that fractional solutions are often obtained, necessitating the need for post-processing techniques such as OSD.

\section{Numerical results}
\label{sec:numerics}
To evaluate the effectiveness of LP+OSD, we use it to decode the rotated surface code, random HGP codes, and BB codes.
We benchmark its performance by comparing it against BP+OSD~\cite{Panteleev2021degeneratequantum} (currently the most common method of decoding general quantum LDPC codes) with the same parameters as in Ref.~\cite{Roffe2020BPOSD}, where the toric code and HGP codes were studied.
We consider the code capacity setting with independent $Z$ errors of probability $p$ and estimate the corresponding logical error rate $p_L$.
Since there is only one type of error present, binary BP is used.
All codes considered are symmetric under exchanging $X$ and $Z$ stabilizers, meaning that there exists
an automorphism of the Tanner graph that permutes $Q$ and maps $C_X$ and $C_Z$ to each other.
This is because we consider square surface codes and HGP codes from two copies of the same classical code. BB codes have an inherent symmetry between $H_X$ and $H_Z$~\cite{bravyi2024high}.
Since under $X$ and $Z$ noise, the codes have equal protection and the decoders work equally well, it is valid to consider only decoding $Z$ errors.
Further details of the numerical simulations are discussed in Appendix~\ref{app:simulationdetails}.
The source code that was used to perform the simulations and the data that was collected are available in Ref.~\cite{zenodocode}.

\begin{figure*}[htpb]
	\centering
	\subfloat[]{\includegraphics[trim={0.35cm 0.4cm 0.3cm 0.35cm},clip,width=0.33\textwidth]{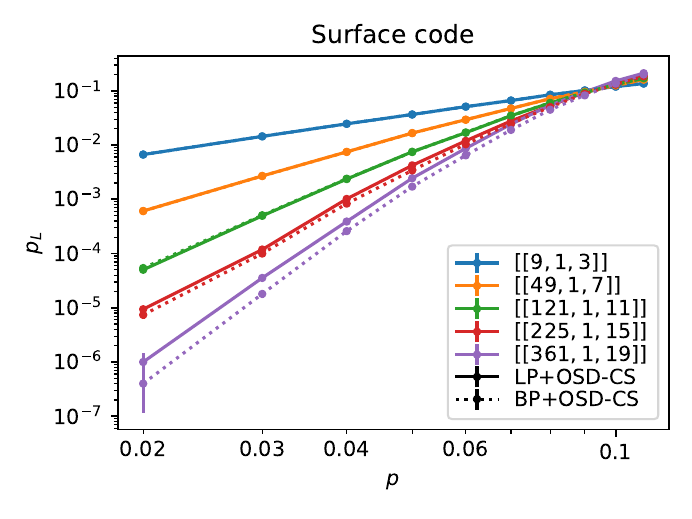}%
		\label{subfig:SC}}
	\hfil
	\subfloat[]{\includegraphics[trim={0.35cm 0.4cm 0.3cm 0.35cm},clip,width=0.33\textwidth]{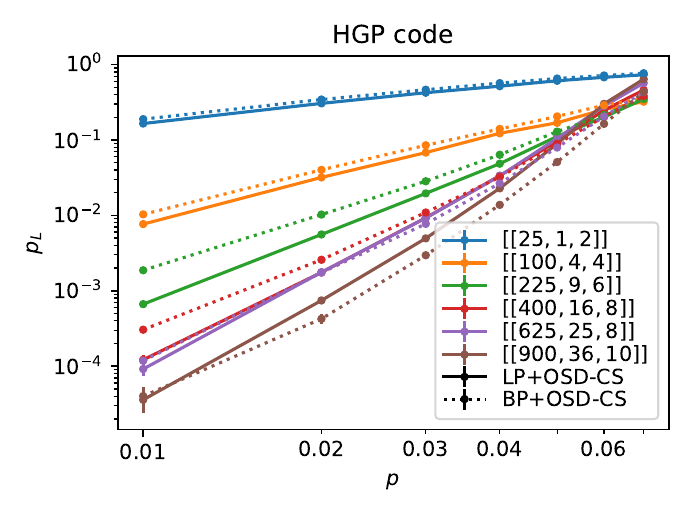}%
		\label{subfig:HGP}}
	\hfil
	\subfloat[]{\includegraphics[trim={0.35cm 0.4cm 0.3cm 0.35cm},clip,width=0.33\textwidth]{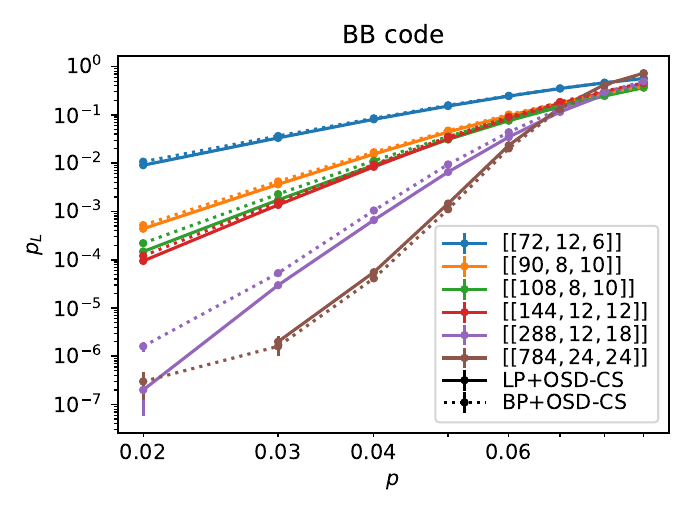}%
		\label{subfig:BB}}
	\caption{Comparison of the LP+OSD-CS decoder (solid lines) with the BP+OSD-CS decoder (dotted lines) for (a) the surface code, (b) random HGP codes, and (c) BB codes.}
	\label{fig:mainresults}
\end{figure*}

\begin{figure*}[htpb]
\centering
\subfloat[]{\includegraphics[trim={0.35cm 0.4cm 0.3cm 0.35cm},clip,width=0.33\textwidth]{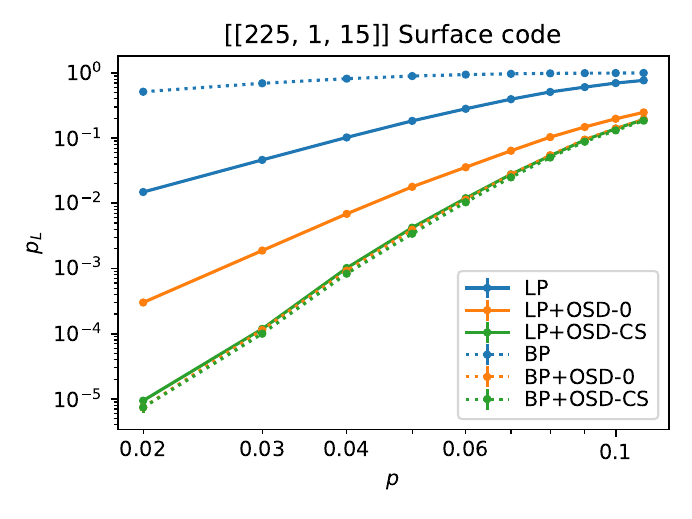}%
	\label{subfig:SC_decodingcomparison}}
\hfil
\subfloat[]{\includegraphics[trim={0.35cm 0.4cm 0.3cm 0.35cm},clip,width=0.33\textwidth]{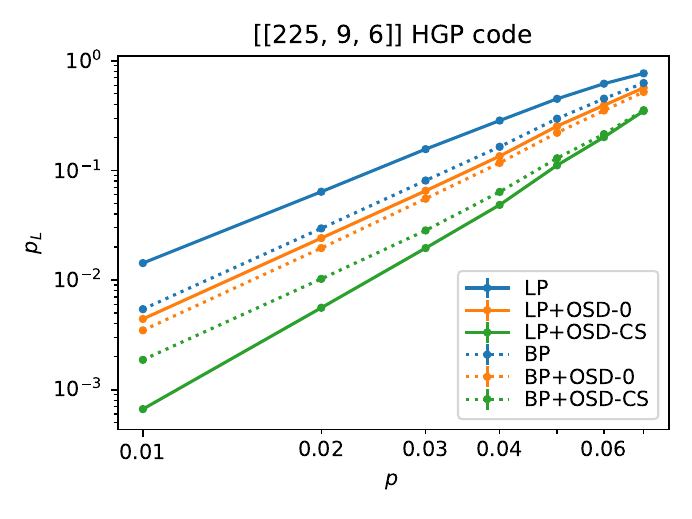}%
	\label{subfig:HGP_decodingcomparison}}
\hfil
\subfloat[]{\includegraphics[trim={0.35cm 0.4cm 0.3cm 0.35cm},clip,width=0.33\textwidth]{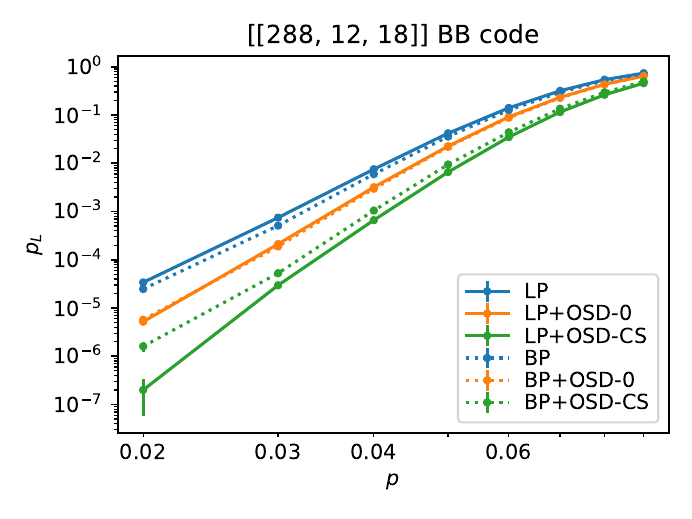}%
	\label{subfig:BB_decodingcomparison}}
\caption{Comparison of LP and BP decoders with different OSD post-processing methods for (a) the $[[225, 1, 15]]$ surface code, (b) random HGP codes with parameters $[[225, \ge 9, \ge 6]]$, and (c) the $[[288, 12, 18]]$ BB code.}
\label{fig:decodingcomparison}
\end{figure*}

The surface code is one of the most well-studied and practical quantum codes with geometrically local stabilizer checks. However, this means that the Tanner graph of the code is not very expanding, and LP by itself is known to perform poorly on it~\cite{FawziLP2021}.

HGP codes are a family of constant-rate quantum LDPC codes which satisfy certain expansion properties when the component classical codes are sufficiently random~\cite{leverrier2015quantum}. LP can perform reasonably well by itself on such codes~\cite{FawziLP2021}.
The HGP codes we use are randomly sampled to have good distance. In particular, they have parameters $[[n=25s^2, k\ge s^2, d\ge d(s)]]$ for $s = 1, 2, 3, 4, 5, 6$ and corresponding lower bounds $d(s) = 2, 4, 6, 8, 8, 10$.

Another notable family of LDPC codes is the recently introduced BB codes~\cite{bravyi2024high}. These codes have weight-six stabilizers that can be implemented using four local interactions—matching the surface code's connectivity—and two additional nonlocal interactions. BB codes demonstrate promising parameters even at small blocklengths and have shown good performance under BP+OSD decoding.
We consider BB codes introduced in Ref.~\cite{bravyi2024high}, which have parameters $[[72, 12, 6]]$, $[[90, 8, 10]]$, $[[108, 8, 10]]$, $[[144, 12, 12]]$, $[[288, 12, 18]]$, and $[[784, 24, 24]]$.

\fig{mainresults} compares the LP+OSD-CS decoder with BP+OSD-CS on the three families of codes. For the surface code, LP+OSD performs comparably to BP+OSD at small code sizes (up to the $[[225, 1, 15]]$ code), while BP+OSD does better as the code size increases. In contrast, LP+OSD consistently outperforms BP+OSD for random HGP codes up to ones with parameters $[[400, \ge 16, \ge 8]]$. It also has slightly lower logical error rates than BP+OSD for BB codes up to the $[[288, 12, 18]]$ code. Small code sizes are the regime where more expensive decoders like BP+OSD or LP+OSD, which have cubic complexity, remain practically viable.

\fig{decodingcomparison} presents a comparison of LP and BP decoding with independent rounding, OSD-0 post-processing, and OSD-CS post-processing for specific codes of 200-300 qubits. The relative performances vary significantly across the surface code, HGP codes, and BB codes. For both LP and BP, OSD-CS does better than OSD-0, which does better than independent rounding. However, the relative improvements from OSD-0 to OSD-CS for LP is greater than those for BP. This may be due to the fact that sometimes LP assigns a value of 0 to a qubit that does have an error (e.g., the error patterns in Sec.~\ref{sec:uncorrectableLPexamples}), which must be searched using a higher-order OSD method.

\begin{figure}[htpb]
	\centering
	\subfloat[]{\includegraphics[trim={0.35cm 0.4cm 0.3cm 0.35cm},clip,width=0.24\textwidth]{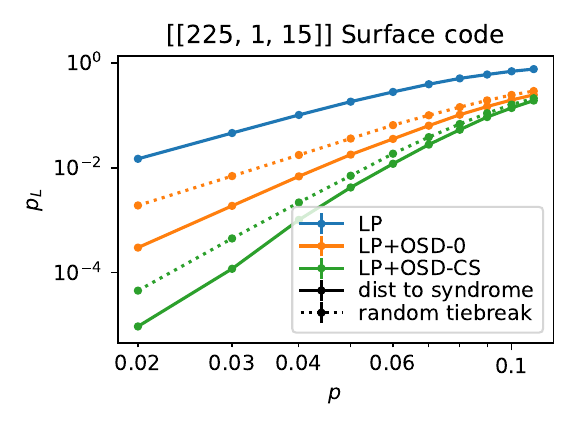}%
		\label{subfig:SC_LPdecodingcomparison}}
	\hfil
	\subfloat[]{\includegraphics[trim={0.35cm 0.4cm 0.3cm 0.35cm},clip,width=0.24\textwidth]{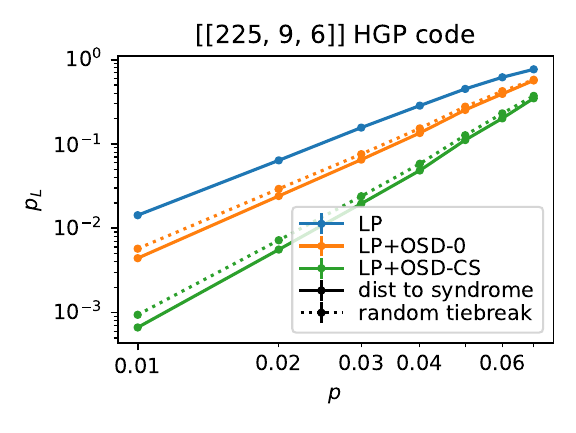}%
		\label{subfig:HGP_LPdecodingcomparison}}
	\hfil
	\subfloat[]{\includegraphics[trim={0.35cm 0.4cm 0.3cm 0.35cm},clip,width=0.24\textwidth]{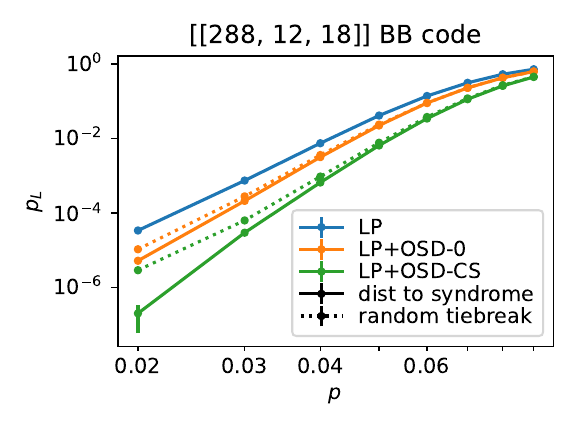}%
		\label{subfig:BB_LPdecodingcomparison}}
	\caption{Comparison of the different LP+OSD decoders. In OSD post-processing, qubits are sorted based on probabilities, with ties broken either based on the distance to a nontrivial syndrome (solid lines) or randomly (dotted lines). We plot the performance for (a) the $[[225, 1, 15]]$ surface code, (b) random HGP codes with parameters $[[225, \ge 9, \ge 6]]$, and (c) the $[[288, 12, 18]]$ BB code.}
	\label{fig:LPdecodingcomparison}
\end{figure}

Recall that in the first step of OSD, the qubits are sorted according to their probability of having an error, and there may be many ties when interpreting an LP solution as probabilities. In \fig{LPdecodingcomparison}, we compare two ways of breaking these ties. The plots show that for the $[[225, 1, 15]]$ surface code, random HGP codes with parameters $[[225, \ge 9, \ge 6]]$, and the $[[288, 12, 18]]$ BB code, breaking ties by distance to a nontrivial syndrome, as described in \secref{LPOSD}, yields lower logical error rates than breaking ties randomly for both OSD-0 and OSD-CS. The same conclusion holds for the other codes in the three families. Therefore, we break ties by distance to a nontrivial syndrome in our other LP+OSD simulations (Figs.~\ref{fig:mainresults} and~\ref{fig:decodingcomparison}).
As expected, LP+OSD-CS outperforms LP+OSD-0, which, in turn, outperforms LP with independent rounding, regardless of the tie-breaking scheme.

Fig.~\ref{fig:wrongsyndromefraction} shows the percentage of times, conditioned on LP failure, that LP with independent rounding gives a solution not consistent with the syndrome of the error. The quantity plotted is $p_{ws}/p_L$, where $p_{ws}$ is the rate of obtaining a rounded solution with the wrong syndrome. The large values of $p_{ws}/p_L$, especially for bigger code sizes, represent the potential gain from using a more sophisticated rounding technique and motivated the exploration of OSD as a post-processing step.

\begin{figure}[htpb]
	\centering
	\subfloat[]{\includegraphics[trim={0.35cm 0.4cm 0.3cm 0.35cm},clip,width=0.24\textwidth]{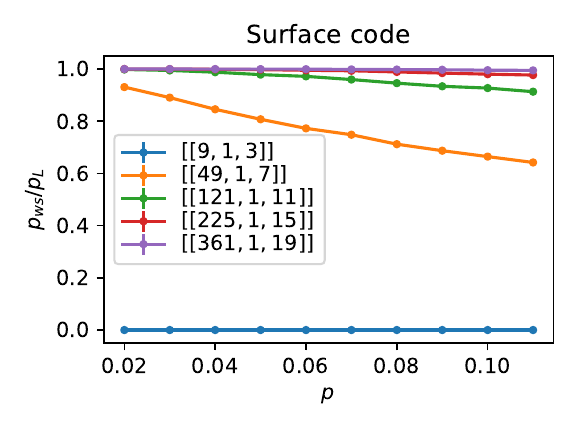}%
		\label{subfig:SC_wrongsyndromefraction}}
	\hfil
	\subfloat[]{\includegraphics[trim={0.35cm 0.4cm 0.3cm 0.35cm},clip,width=0.24\textwidth]{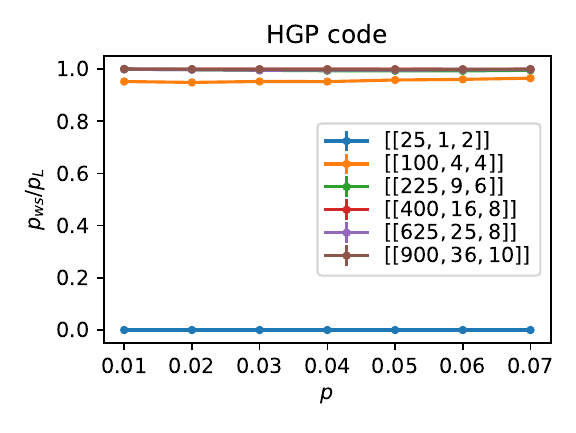}%
		\label{subfig:HGP_wrongsyndromefraction}}
	\hfil
	\subfloat[]{\includegraphics[trim={0.35cm 0.4cm 0.3cm 0.35cm},clip,width=0.24\textwidth]{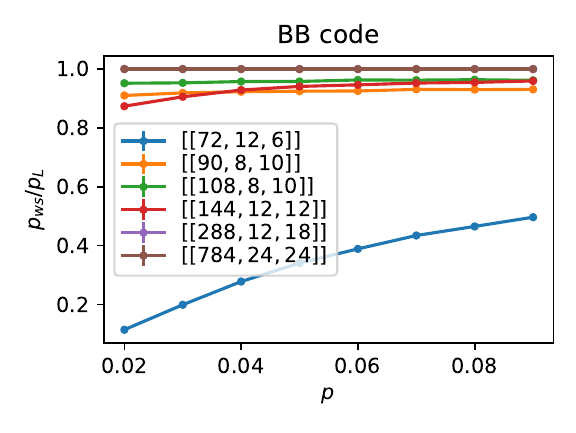}%
		\label{subfig:BB_wrongsyndromefraction}}
	\caption{Fraction of LP failures that are due to a rounded solution not satisfying the syndrome for (a) the surface code, (b) random HGP codes, and (c) BB codes.}
	\label{fig:wrongsyndromefraction}
\end{figure}

Finally, we note that our construction of uncorrectable constant-weight error patterns in Sec.~\ref{sec:uncorrectableLPexamples} helps explain why LP+OSD-CS performs well at small and intermediate code sizes but tends to underperform as we move to larger codes.
The LP solutions to these error patterns have $x_{i_0} = 0$ for a qubit $i_0$ that was actually corrupted.
As mentioned in Section~\ref{subsec:counterexamplesOSD}, when OSD-CS iterates over weight-one configurations of $\tilde e_{\mc T}$, one of the configurations will set qubit $i_0$ to 1. This leads to a successful recovery since $e \subseteq \mc S \cup \{i_0\}$.
However, as the code size increases, the decoder may encounter multiple such problematic patterns simultaneously. When this happens, it is no longer guaranteed that the weight-one or weight-two search performed by OSD-CS will be sufficient to set $\tilde e_{i_0} = 1$ for the corresponding $i_0$ from all the problematic patterns. Consequently, OSD-CS cannot simultaneously set all the relevant qubits to 1, which limits its effectiveness as the code size grows.

\section{Discussion}
\label{sec:discussion}
In this article, we studied LP decoding of quantum LDPC codes. On one hand, we showed that LP with independent rounding runs into problematic error patterns and does not admit a threshold for many families of quantum codes. On the other hand, we proposed the LP+OSD decoder, which incorporates an improved rounding technique and shows promising performance for HGP and BB codes of up to a few hundred qubits.

Based on our findings, many open directions remain to be explored. Theoretically, one could try to find modifications of the LP decoder that allow for a provable threshold in certain quantum codes. For example, adding redundant parity checks or considering multiple checks simultaneously could tighten the LP relaxation to remove certain fractional solutions~\cite{FeldmanCapacity,BurshteinGoldenberg10improvedLP}, although these methods will lead to increased LP problem sizes.
One could also try to find alternative LP relaxations that are aware of the stabilizer generators in a quantum code to obtain provable guarantees.

Practically, heuristic methods to improve the accuracy of LP decoding should also be explored. Taking inspiration from the classical setting, one could consider adaptive LP decoding, which incrementally adds constraints until an integral solution, representing the most likely error, is found~\cite{siegel2008adaptive}. Other possibilities include guessing facets of the LP polytope~\cite{DimakisWainwright06guessingfacets} and branch-and-bound methods~\cite{YangFeldmanWang06nonlinear}. It would be interesting to see what effect these improvements have on LP decoding of quantum codes when combined with OSD post-processing.

In terms of efficiency, LP+OSD has an asymptotic scaling that makes it infeasible for large codes. Several methods have already been developed to reduce the runtime of both LP (in the classical and quantum settings) and OSD (in the context of BP+OSD) without significant degradation in performance. 
The alternating direction method of multipliers (ADMM) reformulates LP decoding as a message-passing algorithm that has runtime comparable to BP but is guaranteed to converge~\cite{barman2013decomposition}, and it has been successful for decoding HGP codes~\cite{groues2022decoding}.
BP can be used as a preprocessor that either fully corrects the error or passes a prior to LP that can reduce its runtime~\cite{wang2009multi}.
Several techniques for reducing the cost of OSD have also been introduced in the context of BP post-processing, such as ordered Tanner forest (OTF)~\cite{iolius2024almost} and localized statistics decoding (LSD)~\cite{hillmann2024localized}, which may be adapted to LP post-processing as well.
Combining these approaches could reduce the runtime of LP+OSD to almost linear complexity, but one would have to evaluate if such techniques can maintain the accuracy of the decoder.

We only benchmarked the performance of LP+OSD in the code capacity setting as a proof-of-concept demonstration. 
For more accurate comparisons with existing decoders, it would be important to determine if, and for what code sizes, our decoder shows an advantage in the more realistic setting of circuit-level depolarizing noise (see Appendix~\ref{app:circuitlevelnoise}).
Circuit-level noise may be problematic for LP+OSD since it gives rise to a decoding problem with more error mechanisms, and LP+OSD seems to perform relatively worse at larger code sizes.
Nevertheless, one may benefit from applying LP decoding directly to depolarizing noise in such a scenario instead of approximating it as independent \(X\) and \(Z\) errors. This approach can potentially capture error correlations more accurately, and one could then compare its performance against quaternary BP.

In conclusion, our work enables further exploration of decoding quantum codes using LP, helping address the scarcity of accurate and efficient decoders for general quantum codes.

\section*{Acknowledgments}
We thank A.~Kubica, A.~Leverrier, C.~Pattison, J.~Preskill, and Y.~Wu for helpful discussions and S.~Puri for providing feedback on an earlier version of this manuscript. S.G. acknowledges funding from the Air Force Office of Scientific Research (FA9550-19-1-0360), the U.S. Department of Energy (DE-SC0018407, DE-AC02-07CH11359), the National Science Foundation (PHY-1733907), and the National Science Foundation's Quantum Leap Challenge Institute in Robust Quantum Simulation (OMA-2120757). M.S. is supported by AWS Quantum Postdoctoral Scholarship and
funding from the National Science Foundation PHY-2317110 and NSF Career Award CCF-2048204. The Institute for Quantum Information and Matter is an NSF Physics Frontiers Center.

{\appendices

\section{An equivalent formulation of LP}\label{app:reformulationLP}
\begingroup
\def\thetheorem{\ref{lem:lpequivalnet}}
\begin{lemma}
	Let $e\in \bb F_2^n$ be a $Z$ error with syndrome $s=H_Xe$ and $e'$ be an error with the same syndrome. Then the LP~\eqref{eq:errorbasedLP}, defined with $e'$, is equivalent to the LP~\eqref{eq:LPRelaxationPrimal1}. More precisely, there is a bijection between feasible solutions given by the reflections
	\begin{align}
		x_i &= \begin{cases}x_i', \quad &i\notin e'\\1-x_i', \quad &i\in e'\end{cases}\quad \forall i\in Q\, ,\\
		w_{j,S} &= w_{j,S + U_j}'\quad \forall j\in C_X,\ S\in E_j^{s_j}\, ,
	\end{align}
	where $x_i$,$w_{j,S}$ are variables of the LP~\eqref{eq:LPRelaxationPrimal1}, $x_i'$,$w_{j,S}'$ are variables of the LP~\eqref{eq:errorbasedLP}, and $U_j = \res{e'}{f_j}$. Two solutions related by the bijection have the same objective values up to an additive constant $|e'|$. Furthermore, the LP decoders based on these solutions give the same output.
\end{lemma}
\addtocounter{theorem}{-1}
\endgroup

\begin{proof}
	Let $x',w'$ be a feasible solution to LP~\eqref{eq:errorbasedLP}, and define $x,w$ as in the lemma. Note that $|U_j| = s_j\mod 2$, so if $S\in E_j^{s_j}$, then $S+U_j\in E_j^0$ and $w_{j,S+U_j}'$ is indeed a variable of the LP~\eqref{eq:errorbasedLP}. We verify that $x,w$ is a feasible solution of the LP~\eqref{eq:LPRelaxationPrimal1}:
	\begin{equation}
		\sum_{S\in E_j^{s_j}}w_{j,S} = \sum_{S\in E_j^{s_j}}w_{j,S+U_j}' = \sum_{\tilde S\in E_j^0}w_{j,\tilde S}' = 1 \quad \forall j\in C_X\, ,
	\end{equation}
	where we made the substitution $\tilde S = S+U_j$. To check the next constraint, let $i\in Q$ and $j\in C_X$ with $\{i, j\}\in E$, and consider two cases. If $i\notin e'$, then for any $S\subseteq f_j$, we have $i\in S$ if and only if $i\in S+U_j$, so
	\begin{equation}
		\sum_{\substack{S\in E_j^{s_j} \\ S\ni i}}w_{j,S} = \sum_{\substack{S\in E_j^{s_j} \\ S\ni i}}w_{j,S+U_j}' = \sum_{\substack{\tilde S\in E_j^0\\\tilde S\ni i}} w_{j,\tilde S}' = x_i' = x_i\, ,
	\end{equation}
	where $\tilde S = S+U_j$ as before. If $i\in e'$, then for any $S\subseteq f_j$, we have $i\in S$ if and only if $i\notin S+U_j$. In this case,
	\begin{align}
		\sum_{\substack{S\in E_j^{s_j} \\ S\ni i}}w_{j,S} &= \sum_{\substack{S\in E_j^{s_j} \\ S\ni i}}w_{j,S+U_j}' = \sum_{\substack{\tilde S\in E_j^0\\\tilde S\not\ni i}} w_{j,S}' \nn\\
		&= \sum_{\tilde S\in E_j^0}w_{j,S}' - \sum_{\substack{\tilde S\in E_j^0\\\tilde S\ni i}} w_{j,S}' = 1 - x_i' = x_i\, ,
	\end{align}
	where we have also used the first constraint. The nonnegativity of $x_i, w_{j,S}$ follow from the fact that the constraints imply $0\le x_i', w_{j, S}'\le 1$. Showing that a feasible solution $x, w$ of the LP~\eqref{eq:LPRelaxationPrimal1} maps to a feasible solution $x', w'$ of the LP~\eqref{eq:errorbasedLP} is similar.
	
	The objective value of the LP~\eqref{eq:LPRelaxationPrimal1} with solution $x$ is
	\begin{equation}
		\sum_{i\in Q} x_i = \sum_{i\notin e} x_i' + \sum_{i\in e'}(1-x_i') = \sum_{i\notin e'} x_i' - \sum_{i\in e}x_i' + |e'|\, ,
	\end{equation}
	which is a constant $|e'|$ greater than that of LP~\eqref{eq:errorbasedLP} with solution $x'$.
	
	The decoder based on LP~\eqref{eq:errorbasedLP} outputs (a rounded version of) $\tilde e$ with $\tilde e_i = x'^*_i$ when $e'_i = 0$ and $\tilde e_i = 1 - x'^*_i$ when $e'_i = 1$. This is the same as the output $x^*$ of the decoder based on LP~\eqref{eq:LPRelaxationPrimal1}.
\end{proof}

\section{Uncorrectable errors in quantum expander codes}\label{app:quantumexpandercodes}
We describe a concrete family of HGP codes that support constant-weight errors uncorrectable by the LP decoder~\eqref{eq:LPRelaxationPrimal1}. These are defined by two classical codes with expanding Tanner graphs, and the resulting HGP code is called a quantum expander code~\cite{leverrier2015quantum}.

\begin{definition}
    A $d_A$-left-regular bipartite graph $G=(A\sqcup B, E)$ is a $(\gamma,\varepsilon)$-left-expander if for any subset $S\subseteq A$ of size $|S|\le \gamma|A|$, its neighborhood has size $|N(S)|\ge (1-\varepsilon)d_A|S|$. A right-expander is defined similarly.
\end{definition}

\begin{theorem}
    \label{thm:HGPcounterexample}
    Let $C$ and $C'$ be classical codes with full-rank parity-check matrices $H\in \bb F_2^{B\times A}$ and $H'\in \bb F_2^{B'\times A'}$, respectively. Suppose that the Tanner graph $T(C)$ is a $(d_A, d_B)$-biregular $(\gamma_B,\varepsilon_B)$-right-expander and that $T(C')$ is a $(d_A',d_B')$-biregular $(\gamma_A',\varepsilon_A')$-left-expander for constants $d_A$, $d_B$, $d_A'$, $d_B'$, $\gamma_B>0$, $\varepsilon_B\le 1/10$, $\gamma_A'>0$, $\varepsilon_A'\le 1/10$ with $d_B + d_A'$ even. Furthermore, suppose that the girths of $T(C)$ and $T(C')$ are greater than eight. Then in the HGP code $\mc C$ constructed from sufficiently large $C$ and $C'$, there exists a constant-weight error that is uncorrectable by the LP decoder~\eqref{eq:LPRelaxationPrimal1}.
\end{theorem}

The theorem will be proven by applying Proposition~\ref{counter:Zcycle} using the constant-weight error described in Sec.~\ref{subsec:uncorrectableerrorsHGP} and illustrated in Fig.~\ref{fig:HPCcounterexample}. The main challenge is to show that the error pattern is reduced, which requires analyzing the structure of $Z$ stabilizers in the code. For a $Z$ stabilizer $g$ that is the sum of generators from a subset $S\subseteq B\times A'$, we define the projections onto the columns $a'\in A'$ and rows $b\in B$ by
\begin{align}
    S_{a'} &= \{b\in B: (b,a')\in S\}\subseteq B\, ,\\
    S_b &= \{a'\in A': (b,a')\in S\}\subseteq A'\, .
\end{align}
Similarly, for an error $e\subseteq A\times A'\sqcup B\times B'$, a column $a'\in A'$, and a row $b\in B$, we define
\begin{align}
    e_{a'} &= \{a\in A: (a,a')\in e\}\subseteq A\, ,\\
    e_b &= \{b'\in B': (b,b')\in e\}\subseteq B'\, .
\end{align}
Recall that we may view $S_{a'}$, $S_b$, $e_{a'}$, and $e_b$ as vectors over $\bb F_2$.

We first show that if a stabilizer reduces the weight of an error in the HGP code $\mc C$, then the error's projection onto at least one row or column must be reduced by the stabilizer's projection onto the same row or column using $H'$ or $H^\top$.

\begin{lemma}
    \label{lem:columnorrowreduction}
    Assume the conditions in Theorem~\ref{thm:HGPcounterexample}. Let $e$ be a $Z$ error and $g$ be a $Z$ stabilizer defined by a subset $S\subseteq B\times A'$ such that $|e+g|<|e|$. Then there exists $a'\in A'$ such that $|e_{a'} + H^\top S_{a'}|<|e_{a'}|$ or there exists $b\in B$ such that $|e_b+H'S_b|<|e_b|$.
\end{lemma}
\begin{proof}
    The support of $e$ can be decomposed as
    \begin{equation}
        e = \bigsqcup_{a'\in A'} e_{a'}\times \{a'\} \sqcup \bigsqcup_{b\in B} \{b\}\times e_b\, ,
    \end{equation}
    while the support of $g$ can be decomposed as
    \begin{equation}
        g = \bigsqcup_{a'\in A'} H^\top S_{a'}\times \{a'\} \sqcup \bigsqcup_{b\in B} \{b\}\times HS_b\, .
    \end{equation}
    Therefore, the initial error weight is
    \begin{equation}
        |e| = \sum_{a'\in A'} |e_{a'}| + \sum_{b\in B} |e_b|\, ,
    \end{equation}
    and its weight after adding $g$ is
    \begin{equation}
        |e + g| = \sum_{a'\in A'} |e_{a'} + H^\top S_{a'}| + \sum_{b\in B} |e_b + H'S_b|\, .
    \end{equation}
    The result then follows because one of the summands of $|e+g|$ must be less than the corresponding summand of $|e|$.
\end{proof}

The above lemma allows us to consider individual rows or columns when determining if an error is reduced. We now show that if the support of $S$ is too large on a row or column, the resulting stabilizer will have large support.

\begin{lemma}
    \label{lem:projectionweight}
    Assume the conditions in Theorem~\ref{thm:HGPcounterexample} and let $g$ be a $Z$ stabilizer defined by a subset $S\subseteq B\times A'$. If $|S_b|\ge \gamma_A'|A'|$ for some $b\in B$, then $|g|\ge \gamma_A'|A'|$. If $|S_{a'}| \ge \gamma_B|B|$ for some $a'\in A'$, then $|g|\ge \min\{\gamma_B|B|, \gamma_A'|A'|\}$.
\end{lemma}
\begin{proof}
    If $|S_b|\ge \gamma_A'|A'|$, there are at least $\gamma_A'|A'|$ values of $a'\in A'$ such that $S_{a'}\ne \emptyset$. Because $H$ has full rank, $g$ will have support on the corresponding columns $A\times \{a'\}\subseteq A\times A'$. This shows $|g|\ge \gamma_A'|A'|$.

    Now suppose $|S_{a'}| \ge \gamma_B|B|$. There are at least $\gamma_B|B|$ values of $b\in B$ such that $S_b\ne \emptyset$. If for all of these $b$, the stabilizer $g$ has nontrivial support on the corresponding row $\{b\}\times B'\subseteq B\times B'$, then $|g|\ge \gamma_B|B|$. Otherwise, there is at least one $b\in B$ where $S_b\ne \emptyset$ and $H'S_b=0$, where we view $S_b\in \bb F_2^{A'}$ as a codeword of $C'$. The expansion of $T(C')$ guarantees that the distance of $C'$ is at least $\gamma_A'|A'|$~\cite{SipserSpielman96}. Therefore, $|S_{a'}|\ge \gamma_A'|A'|$, and we may use the previous analysis. Combining both cases gives $|g|\ge \min\{\gamma_B|B|, \gamma_A'|A'|\}$.
\end{proof}

We are now ready to prove the theorem.

\begin{proof}[Proof of Theorem~\ref{thm:HGPcounterexample}]
    Because the girths of $T(C)$ and $T(C')$ are greater than eight, we may define an error pattern in Fig.~\ref{fig:HPCcounterexample} from the cycle~\eqref{eq:HGPcycle}, which we label $\Gamma$. Recall that $\Gamma$ is formed by concatenating the paths $P\times a_1'$, $b_3\times P'$, $P\times a_3'$, and $b_1\times P'$ in $T(\mc C)$, where $P = b_1,a_1,b_2,a_2,b_3$ and $P' = a_1',b_1',a_2',b_2',a_3'$ are paths in the $T(C)$ and $T(C')$, respectively. Let $R = \{b_1, b_2, b_3\}\times \{a_1', a_2', a_3'\}\setminus \{(b_2, a_2')\}$ be the $Z$ generators in the cycle and $I=\{a_1, a_2\}\times\{a_1',a_3'\}\sqcup\{b_1,b_3\}\times\{b_1',b_2'\}$ be the intersections of their supports. We consider the constant-weight $Z$ error $e$ defined by corrupting $\{a_1,a_1'\}$ and half the qubits in $g_{(b,a')}\setminus I$ for each $(b,a')\in R$. We choose $e$ such that half of $g_{(b,a')}\setminus I\cap A\times A'$ and half of $g_{(b,a')}\setminus I\cap B\times B'$, each rounded up or down appropriately, are corrupted, but the qubits $(a_1, a_2')$, $(a_2, a_2')$, $(b_2, b_1')$, and $(b_2, b_2')$ are not corrupted.
    From the discussion in Sec.~\ref{subsec:uncorrectableerrorsHGP}, all the conditions of Proposition~\ref{counter:Zcycle} other than $e$ being reduced are satisfied, so it remains to show that no $Z$ stabilizer can reduce the weight of $e$.

    Suppose for a contradiction that $g$ is a $Z$ stabilizer defined by a subset $S\subseteq B\times A'$ such that $|e+g|<|e|$.
    Note that adding generators from any subset of $R$ does not decrease the weight of the error. In fact, the resulting operator will have the same form as described in the previous paragraph, with the modification that it could have larger support on $I$. 
    Therefore, we may assume without loss of generality that $S\cap R = \emptyset$.    
    By Lemma~\ref{lem:columnorrowreduction}, there exists $a'\in A'$ such that $|e_{a'} + H^\top S_{a'}|<|e_{a'}|$ or there exists $b\in B$ such that $|e_b+H'S_b|<|e_b|$. Suppose it is the former, with a similar analysis for the latter case.
    
    Consider the weight of $S_{a'}$. If $|S_{a'}|\ge \gamma_B|B|$, then Lemma~\ref{lem:projectionweight} implies that $|g|\ge \min\{\gamma_B|B|, \gamma_A'|A'|\}$, which is linear in the size of $C$ and $C'$. But this is not possible for sufficiently large codes because $|e+g| - |e|\ge |g| - 2|e| > 0$, as $|e|$ is constant.
    
    If $|S_{a'}|< \gamma_B|B|$, the expansion of $T(C)$ implies that $|N(S_{a'})|\ge (1-\varepsilon_B)d_B|S_{a'}|$. In fact, the number of unique neighbors of $S_{a'}$---those that are adjacent to only one vertex in $S_{a'}$---is at least $(1-2\varepsilon_B)d_B|S_{a'}|$. Thus, $|H^\top S_{a'}|\ge (1-2\varepsilon_B)d_B|S_{a'}|$, and we have
    \begin{align}
        \label{eq:expansionSasize}
        (1-2\varepsilon_B)d_B|S_{a'}| \le |H^\top S_{a'}| &\le |e_{a'} + H^\top S_{a'}| + |e_{a'}|\nn\\
        &< 2|e_{a'}|\, .
    \end{align}
    Note that $d_B\ge 1/\varepsilon_B\ge 10$; otherwise, the intersection of a cycle in $T(C)$ with $B$ would be a set of insufficient expansion. Also, our definition of the error means that its restriction to $A\times \{a'\}$ has weight
    \begin{equation}
        |e_{a'}|\le \left\lceil \frac{d_B-1}{2} \right\rceil + \left\lceil \frac{d_B-2}{2} \right\rceil + \left\lceil \frac{d_B-1}{2}  \right\rceil + 2 \le \frac{3}{2}d_B + 1
    \end{equation}
    if $a'\in \{a_1', a_3'\}$ (where the first three terms correspond to errors in $(H^\top b_i\setminus \{a_1, a_2\}) \times \{a'\}$ and the $+2$ accounts for possible errors in $\{a_1,a_2\}\times \{a'\}$) and weight
    \begin{equation}
        |e_{a'}|\le 2\left\lceil \frac{d_B}{2} \right\rceil \le d_B + 1
    \end{equation}
    if $a' = a_2'$. Other values of $a'$ are not possible since the error has no support on other columns. In either case, rearranging Eq.~\eqref{eq:expansionSasize} gives
    \begin{equation}
        |S_{a'}| < \frac{2(\frac{3}{2}d_B+1)}{(1-2\varepsilon_B)d_B} = \frac{3}{1-2\varepsilon_B} + \frac{2}{(1-2\varepsilon_B)d_B} \le 4\, .
    \end{equation}
    
    Now, we use the condition that adding $H^\top S_{a'}$ reduces the weight of $e_{a'}$. We can assume that $S_{a'}\sqcup H^\top S_{a'}$ is a connected subset of vertices in $T(C)$ because if not, one of its connected components $\tilde S_{a'}\sqcup H^\top \tilde S_{a'}$ would also satisfy $|e_{a'} + H^\top \tilde S_{a'}|<|e_{a'}|$ and we could apply the reasoning below to the subset $\tilde S_{a'}\subseteq S_{a'}$ instead.
    The condition $|e_{a'} + H^\top S_{a'}|<|e_{a'}|$ and the identity $|e_{a'} + H^\top S_{a'}| = |e_{a'}| + |H^\top S_{a'}| - 2|H^\top S_{a'}\cap e_{a'}|$ imply $|H^\top S_{a'}\cap e_{a'}|\ge \frac{1}{2}|H^\top S_{a'}|$. In particular, $\frac{1}{2}|H^\top S_{a'}|\ge 5$ from $d_B\ge 10$ and the expansion of $T(C)$. Since $e_{a'} \subseteq N(b_1)\cup N(b_2) \cup N(b_3)$, the pigeonhole principle implies that $|H^\top S_{a'}\cap N(b_i)| \ge 2$ for some $i\in \{1,2,3\}$. Let $a_3, a_4\in H^\top S_{a'}\cap N(b_i)$ be two distinct vertices in $T(C)$.
    The path from $a_3$ to $a_4$ in $S_{a'}\sqcup H^\top S_{a'}$ has length at most six since $|S_{a'}| \le 3$, and it does not include $b_i$ since we assumed $S\cap R=\emptyset$. But this means the concatenation of that path with $a_3,b_i,a_4$ is a cycle of length at most eight in $T(C)$, contradicting its girth. Therefore, no $Z$ stabilizer can reduce the weight of $e$.    
\end{proof}

We have not optimized the constants in Theorem~\ref{thm:HGPcounterexample}, but doing so would likely show uncorrectable error patterns in HGP codes from classical codes with less expansion and lower girth. Although the proof technique does not directly apply to non-expanding codes, we expect that other methods can show similar error patterns to be reduced and be problematic for the LP decoder.

\section{LP decoding of circuit-level noise}
\label{app:circuitlevelnoise}
In this appendix, we describe how LP+OSD can handle more general noise models such as phenomenological or circuit-level noise. Since code capacity or phenomenological noise can be mapped to circuit-level noise with errors only at specific locations, we will discuss decoding a stabilizer circuit with Pauli noise.

Given a circuit, we first need to identify the \emph{error mechanisms} affecting the circuit and define \emph{detectors}~\cite{Gidney2021stim}. An error mechanism $v_i$ is a Pauli error that occurs at a specific location in the circuit with probability $p_i$. For example, $v_i$ could be a two-qubit Pauli error after a two-qubit operation, or it could be a measurement error which is a bit-flip error on the classical outcome. We assume that error mechanisms occur independently, so the total distribution of errors in the circuit is a product of all the error mechanisms. In this framework, an $r$-qubit depolarizing channel should be decomposed as $4^r-1$ independent error mechanisms corresponding to the nontrivial $r$-qubit Pauli operators. A detector $D_j$ is a product of measurement outcomes in the circuit that is deterministic in the absence of errors but reveals information about possible errors when flipped. For syndrome extraction circuits of stabilizer codes, the detectors are often chosen to be the product of consecutive stabilizer measurements.

We then define the \emph{detector check matrix}~\cite{Higgott2025sparseblossom} $H$ to be the binary matrix with rows corresponding to the different detectors and columns corresponding to the different error mechanisms such that $H_{ji}=1$ exactly when $v_i$ occurring causes the value of $D_j$ to flip.

Finding the most likely error that has affected the circuit is equivalent to decoding the classical code with parity-check matrix $H$. Thus, we can use the LP+OSD decoder. One modification that may be needed is in the case when the error mechanisms do not all occur with the same probability, as we had assumed in the main text. In this case, we would change the coefficients of the objective function in the LP~\eqref{eq:LPRelaxationPrimal1} so that we minimize $\sum_{i}w_ix_i$ instead of $\sum_{i}x_i$, where $w_i = \log{(1-p_i)/p_i}$ is the log probability ratios of the error mechanism $v_i$. The OSD post-processing step is unchanged.

Note that as with stabilizers in the code capacity setting, there are combinations of error mechanisms that are harmless for circuit-level noise. For example, a qubit may suffer two consecutive $Z$ errors, with all measurements in between that would be flipped also suffering a measurement error. However, since we are only interested in finding the most likely error and our decoder does not make use of the degeneracy, these harmless configurations of errors do not need to be explicitly calculated.

\section{Simulation details}
\label{app:simulationdetails}
We give more details on the simulations in Sec.~\ref{sec:numerics}.
We applied Monte Carlo sampling to study the performance of the various decoders on the rotated surface code, random HGP codes, and BB codes.
The random HGP codes are sampled as follows. We create a uniformly random connected $(3,4)$-biregular bipartite graph $G$ with $|A|=4s$ left vertices and $|B|=3s$ right vertices for a parameter $s$. Consider $G$ as the Tanner graph of classical codes $C$ with parity-check matrix $H\in \bb F_2^{B\times A}$ and $C^\top$ with parity-check matrix $H^\top\in \bb F_2^{A\times B}$. We postselect on both $C$ and $C^\top$ having distance at least $d(s) = 2, 4, 6, 8, 8, 10$ for $s = 1, 2, 3, 4, 5, 6$, respectively (where the distance of the trivial code $\{0\}$ is defined to be infinity). The HGP code $\mc C$ is the hypergraph product of $C$ with itself, which has parameters $[[n=25s^2, k\ge s^2, d\ge d(s)]]$. Parameters and the maximum stabilizer generator weights of all the codes used are shown in Table~\ref{tab:codeparameters}.

\renewcommand{\arraystretch}{1.2}
\setlength\tabcolsep{0.5em}
\begin{table}[htpb]
    \centering
    \caption{Parameters of the codes used in the simulations.}
    \label{tab:codeparameters}
    \begin{tabular}{|c|c|c|c|c|c|}
    \hline
    \textbf{Code} & $\boldsymbol{n}$ & $\boldsymbol{k}$ & $\boldsymbol{d}$ & $\boldsymbol{k/n}$ & \textbf{Max generator weight} \\
    \hline
    \multirow{5}{*}{Surface code}
     & 9   & 1 & 3  & 0.111 & 4 \\
     & 49  & 1 & 7  & 0.020 & 4 \\
     & 121 & 1 & 11 & 0.008 & 4 \\
     & 225 & 1 & 15 & 0.004 & 4 \\
     & 361 & 1 & 19 & 0.003 & 4 \\
    \hline
    \multirow{6}{*}{HGP code}
     & 25  & $\geq 1$  & $\geq 2$  & 0.040 & 7 \\
     & 100 & $\geq 4$  & $\geq 4$  & 0.040 & 7 \\
     & 225 & $\geq 9$  & $\geq 6$  & 0.040 & 7 \\
     & 400 & $\geq 16$ & $\geq 8$  & 0.040 & 7 \\
     & 625 & $\geq 25$ & $\geq 8$  & 0.040 & 7 \\
     & 900 & $\geq 36$ & $\geq 10$ & 0.040 & 7 \\
    \hline
    \multirow{6}{*}{BB code}
     & 72  & 12 & 6  & 0.167 & 6 \\
     & 90  & 8  & 10 & 0.089 & 6 \\
     & 108 & 8  & 10 & 0.074 & 6 \\
     & 144 & 12 & 12 & 0.083 & 6 \\
     & 288 & 12 & 18 & 0.042 & 6 \\
     & 784 & 24 & 24 & 0.031 & 6 \\
    \hline
    \end{tabular}
\end{table}

To estimate the logical error rate $p_L$, we apply $Z$ errors independently to each qubit of the code with probability $p$, giving an error $e$. The syndrome $s=H_Xe$ is passed to the decoder, which outputs a correction $\tilde e$. By repeating multiple times, we estimate the logical error rate $p_L$ as the fraction of times $e + \tilde e$ is not a $Z$ stabilizer. For Fig.~\ref{fig:wrongsyndromefraction}, we instead determine the fraction of times $H_X\tilde e\ne s$, considering only the trials where LP with independent rounding fails.

For the surface code and BB codes, we took between 50,000 and 10 million samples for each value of $p$ to obtain reasonable error bars. For the random HGP codes, each data point was the result of randomly sampling between 1000 and 60,000 codes, decoding 10 different errors for each code, and taking the total fraction over all codes and decoding trials. The error bars for this procedure were obtained via bootstrapping.

The decoders considered were LP or BP with independent rounding, OSD-0 post-processing, and OSD-CS post-processing. For these decoders, we used the same parameters as in Ref.~\cite{Roffe2020BPOSD}: min-sum BP with variable scaling factor $\alpha = 1 - 2^{-t}$, where $t$ is the iteration number, and $\lambda=60$ in OSD-CS. The BP-based decoders were implemented using the LDPC package~\cite{Roffe_LDPC_Python_tools_2022}. The LPs were solved using Gurobi~\cite{gurobi}. See Ref.~\cite{zenodocode} for the source code for the simulations and the data used to generate the plots.

}

\bibliographystyle{IEEEtran}
\bibliography{mehdis.bib}

\end{document}